\documentclass{amsart}
\usepackage{amsmath,amssymb,amsthm,latexsym}
\usepackage{graphicx}
\usepackage{stmaryrd}		
\usepackage{xcolor}
\usepackage{hyperref}

\newtheorem{theorem}{Theorem}[section]
\newtheorem{corollary}[theorem]{Corollary}
\newtheorem{lemma}[theorem]{Lemma}
\newtheorem{proposition}[theorem]{Proposition}
\newtheorem{question}[theorem]{Question}
\theoremstyle{definition}
\newtheorem{definition}[theorem]{Definition}
\newtheorem{example}[theorem]{Example}
\newtheorem{remark}[theorem]{Remark}

\newcommand{\F}{\mathbb F}
\newcommand{\I}{\mathbb I}
\newcommand{\N}{\mathbb N}
\newcommand{\Z}{\mathbb Z}
\newcommand{\ST}{\operatorname*{ST}}
\newcommand{\colonequal}{\mathrel{\mathop:}=}
\newcommand{\nequiv}{\mathrel{\not\equiv}}

\thanks{This project has received funding from the European Research Council (ERC) under the European Union's Horizon 2020 
research and innovation programme under the Grant Agreement No 648132. }
\begin{document}

\title[Automaticity and invariant measures of linear cellular automata]{Automaticity and invariant measures \\ of linear cellular automata}

\author{Eric Rowland}
\address{
	Department of Mathematics \\
	Hofstra University \\
	Hempstead, NY \\
	USA
}

\author{Reem Yassawi}
\address{
Institut Camille Jordan, Universit\'{e} Lyon-1, France;
School of Mathematics and Statistics, Open University, UK
}


\date{February 19, 2020}

\begin{abstract}
We show that spacetime diagrams of linear cellular automata $\Phi : \F_p^\Z \to \F_p^\Z$ with $(-p)$-automatic initial conditions are automatic.
This extends existing results on initial conditions which are eventually constant.
Each automatic spacetime diagram defines a $(\sigma, \Phi)$-invariant subset of $\F_p^\Z$, where $\sigma$ is the left shift map, and if the initial condition is not eventually periodic then this invariant set is nontrivial.
For the Ledrappier cellular automaton we construct a family of nontrivial $(\sigma, \Phi)$-invariant measures on $\F_3^\Z$.
Finally, given a linear cellular automaton $\Phi$, we construct a nontrivial $(\sigma, \Phi)$-invariant measure on $\F_p^\Z$ for all but finitely many $p$.
\end{abstract}

\maketitle

\section{Introduction}\label{Introduction}

In this article, we study the relationship between $p$-automatic sequences and spacetime diagrams of linear cellular automata over the finite field $\F_p$, where $p$ is prime.
For definitions, see Section~\ref{Preliminaries}.

There are many characterisations of $p$-automatic sequences.
For readers familiar with substitutions, Cobham's theorem~\cite{Cobham-1972} tells us that they are codings of fixed points of length-$p$ substitutions.
In an algebraic setting, Christol's theorem tells us that they are precisely those sequences whose generating functions are algebraic over $\F_p(x)$.
In \cite{Rowland--Yassawi:2015}, we characterise $p$-automatic sequences as those sequences that occur as columns of two-dimensional spacetime diagrams of linear cellular automata $\Phi : \F_p^\Z \to \F_p^\Z$, starting with an eventually periodic initial condition.

We investigate the nature of a spacetime diagram as a function of its initial condition, when the initial condition is $p$-automatic.
In the special case when the initial condition is eventually $0$ in both directions and the cellular automaton has right radius $0$, this question has been thoroughly studied in a series of articles by Allouche, von Haeseler, Lange, Petersen, Peitgen, and Skordev~\cite{AvHLPS:1997, AvHPS:1997, AvHPS:1996}.
Amongst other things, the authors show that an $\N \times \N$-configuration which is generated by a linear cellular automaton, whose right radius is $0$, and an eventually $0$ initial condition, is {\em $[p,p]$-automatic}.
In \cite{Pivato--Yassawi-2008}, Pivato and the second author have also studied the substitutional nature of spacetime diagrams of more general cellular automata with eventually periodic initial conditions.

In Sections~\ref{Automaticity and algebraicity of spacetime diagrams} and \ref{Automaticity of Z x Z diagrams} we extend these previous results by relaxing the constraints imposed on the initial conditions and the cellular automata.
We allow initial conditions to be bi-infinite $(-p)$-automatic sequences or, equivalently, concatenations of two $p$-automatic sequences.
Iterating $\Phi$ produces a $\Z\times \N$-configuration, and we show in Theorem~\ref{sheared STD is automatic}, Theorem~\ref{cone STD is automatic}, and Corollary~\ref{STD is automatic}, that such spacetime diagrams are automatic, with two possible definitions of automaticity: either by {\em shearing} a configuration supported on a cone or by considering $[-p,p]$-automaticity.
Our results are constructive, in that given an automaton that generates an automatic initial condition, we can compute an automaton that generates the spacetime diagram.
We perform such a computation in Example~\ref{p=3 example}, which we use as a running example throughout the article.
While the spacetime diagram has a substitutional nature, the alphabet size makes the computation of this substitution by hand infeasible, and indeed difficult even using software.

We can also extend a spacetime diagram backward in time to obtain a $\Z \times \Z$-configuration where each row is the image of the previous row under the action of the cellular automaton.
In Lemma~\ref{initial conditions for LCA} we show that the initial conditions that generate a $\Z \times \Z$-configuration are supported on a finite collection of lines.
In Theorem~\ref{2D-STDs_are_automatic}, we show that if the initial conditions are chosen to be $p$-automatic, then the resulting spacetime diagram is a concatenation of four $[p,p]$-automatic configurations.

Apart from the intrinsic interest of studying automaticity of spacetime diagrams, one motivation for our study is a search for closed nontrivial sets in $\F_p^\Z$ which are invariant under the action of both the left shift map $\sigma$ and a fixed linear cellular automaton $\Phi$.
Analogously, we also search for measures $\mu$ on one-dimensional subshifts $(X,\sigma)$ that are invariant under the action of both $\sigma$ and $\Phi$.

We give a brief background.
Furstenberg~\cite{Furstenberg-1967} showed that any closed subset of the unit interval $I$ which is invariant under both maps $x \mapsto 2x \bmod 1$ and $x \mapsto 3x \bmod 1$ must be either $I$ or finite.
This is an example of {\em topological rigidity}.
Furstenberg asked if there also exists a {\em measure rigidity}, i.e.\ if there exists a nontrivial measure $\mu$ on $I$ which is invariant under these same two maps.
By ``nontrivial'' we mean that $\mu$ is neither the Lebesgue measure nor finitely supported.
This question is known as the $({\times 2}, {\times 3})$ problem.

The $({\times 2}, {\times 3})$ problem has a symbolic interpretation, which is to find a measure on $\F_2^\N$ which is invariant under both $\sigma$, which corresponds to multiplication by $2$, and the map $u\mapsto u+\sigma (u)$, which corresponds to multiplication by $3$ and where $+$ represents addition with carry.
One can ask a similar question for $\sigma$ and the {\em Ledrappier} cellular automaton $u\mapsto u+\sigma (u)$, where $+$ represents coordinate-wise addition modulo $2$.
One way to produce such measures is to average iterates, under the cellular automaton, of a shift-invariant measure, and to take a limit measure.
Pivato and the second author~\cite{Pivato--Yassawi:2002} show that starting with a Markov measure, this procedure only yields the Haar measure $\lambda$. Host, Maass, and Martinez~\cite{Host--Maass--Martinez:2003} show that if a $(\sigma,\Phi)$-invariant measure has positive entropy for $\Phi$ and is ergodic for the shift or the $\Z^2$-action then $\mu=\lambda$. The problem of identifying which measures are $(\sigma, \Phi)$-invariant is an open problem; see for example Boyle's survey article~\cite[Section 14]{Boyle-2008} on open problems in symbolic dynamics or Pivato's article~\cite[Section 3]{Pivato-2012} on the ergodic theory of cellular automata.

In Sections~\ref{dynamics} and \ref{Invariant measures for linear cellular automata} we apply results of Sections~\ref{Automaticity and algebraicity of spacetime diagrams} and \ref{Automaticity of Z x Z diagrams} to find $(\sigma, \Phi)$-invariant sets and measures.
Spacetime diagrams generate subshifts $(X, \sigma_1, \sigma_2)$, where $\sigma_1$ and $\sigma_2$ are the left and down shifts, and these subshifts project to closed sets in $\F_p^\Z$ that are $(\sigma,\Phi)$-invariant. Similarly, we show in Proposition~\ref{not Lebesgue} that $(\sigma_1, \sigma_2)$-invariant measures on $X$ project to $(\sigma,\Phi)$-invariant measures supported on a subset of $\F_p^\Z$.
Einsiedler~\cite{Einsiedler-2004} constructs, for each $s$ in the interval $0 \leq s \leq 1$, a $(\sigma_1,\sigma_2)$-invariant set and a $(\sigma_1,\sigma_2)$-invariant measure whose entropy in any direction is $s$ times the maximal entropy in that direction. He builds invariant sets using {\em intersection sets} as described in Section~\ref{invariant sets} and asks if every $(\sigma_1,\sigma_2)$-invariant set is an intersection set. He also asks for the nature of the invariant measures. We show in Theorem~\ref{invariant} that each automatic spacetime diagram generates a $(\sigma,\Phi)$-invariant set of small (one-dimensional) word complexity. If we assume that the initial condition is not spatially periodic and the cellular automaton is not a shift, we show in Proposition~\ref{nonperiodic} that these sets are nontrivial.
The invariant sets we find are not obviously intersection sets.

The quest for nontrivial $(\sigma, \Phi)$-invariant measures appears to be more delicate.
Let $(X_U, \sigma_1, \sigma_2)$ be a subshift generated by a $[-p, -p]$-automatic configuration $U$.
Theorem~\ref{nature of measures} states that the measures supported on such subshifts are convex combinations of measures supported on codings of substitution shifts.
We show in Theorem~\ref{complexity} that $U$ has at most polynomial complexity.
Therefore the $(\sigma, \Phi$)-invariant measures guaranteed by Proposition~\ref{not Lebesgue} are not the Haar measure.
However they may be finitely supported:
the shift $X_U$ generated by a nonperiodic spacetime diagram $U$ can contain periodic points on which a shift-invariant measure is supported.
In Theorems~\ref{coincidence} and \ref{unique letter at nonzero coefficients} we identify cellular automata and nonperiodic initial conditions that yield two-dimensional shifts containing constant configurations.
 
We show in Corollary~\ref{not point mass} that spacetime diagrams that do not contain large one-dimensional repetitions support nontrivial $(\sigma,\Phi)$-invariant measures, and in Theorem~\ref{decidability} we show that this condition is decidable.
In Theorem~\ref{power-free} we show that for the Ledrappier cellular automaton there exists a family of substitutions all of whose spacetime diagrams, including our running example, support nontrivial measures.
In Theorem~\ref{power-free general}, we generalise this last proof, showing that for any linear cellular automaton $\Phi$, nontrivial $(\sigma,\Phi)$-invariant measures exist for all but finitely many primes $p$.
Given $\Phi:\F_p^\Z\rightarrow \F_p^\Z$, to what extent it is the case that a random $p$-automatic initial condition generates a nontrivial $(\sigma, \Phi)$-invariant measure?
This remains open.

We are indebted to Allouche and Shallit's classical text~\cite{Allouche--Shallit:2003}, referring to proofs therein on many occasions, which carry through in our extended setting.
In Section~\ref{Preliminaries}, we provide a brief background on linear cellular automata, larger rings of generating functions in two variables, and $p$- and $(-p)$-automaticity. In Section~\ref{Automaticity and algebraicity of spacetime diagrams} we prove that $\Z\times \N$-indexed spacetime diagrams are automatic if we start with automatic initial conditions.
In Section~\ref{Automaticity of Z x Z diagrams} we extend these results to include $\Z \times \Z$-indexed spacetime diagrams.
In Section~\ref{dynamics} we show that automatic spacetime diagrams for $\Phi$ yield nontrivial $(\sigma, \Phi)$-invariant sets and discuss their relation to the intersection sets defined by Kitchens and Schmidt~\cite{Kitchens--Schmidt-1992}.
Finally in Section~\ref{Invariant measures for linear cellular automata}, we study $(\sigma, \Phi)$-invariant measures supported on automatic spacetime diagrams.

\section{Preliminaries}\label{Preliminaries}

\subsection{Linear cellular automata}\label{Linear cellular automata}

Let $\mathcal A$ be a finite alphabet. An element in $\mathcal A^\Z$ is called a {\em configuration} and is written $u=(u_m)_{m \in \Z}$.
 The
 (left) {\em shift map} $\sigma:
 \mathcal A^\Z \rightarrow
 \mathcal A^\Z$ is
 the map defined as $(\sigma(u))_m\colonequal u_{m+1}$. Let $\mathcal A$
 be endowed with the
 discrete topology and $\mathcal A^\Z$ with the product
 topology;
 then $\mathcal A^\Z$ is a metrisable Cantor
 space.
 A (one-dimensional) {\em cellular
automaton} is a continuous, $\sigma$-commuting map 
$\Phi: \mathcal A^\Z\rightarrow \mathcal A^\Z$. The Curtis--Hedlund--Lyndon theorem tells us that a cellular automaton is determined by a local rule $f$: there exist integers $\ell$ and $r$ with $-\ell \leq r$ and $f : \mathcal A^{r + \ell + 1} \to \mathcal A$ such that, for all $m\in \Z$, $(\Phi(u))_m=f(u_{m - \ell}, \dots, u_{m + r})$.
 Let $\N$ denote the set of non-negative integers.

\begin{definition}\label{spacetime}
Let $\Phi:\mathcal{A}^\Z\rightarrow \mathcal{A}^\Z$ be a cellular automaton and let $u\in \mathcal A^\Z$.
If $U\in \mathcal A^{\Z \times\N }$ satisfies $U|_{\Z \times \{0\}} = u$ and $\Phi(U|_{\Z\times \{n\}}) = U|_{\Z\times \{n + 1\}}$ for each $n\in \N$, we call $U=\ST_{\Phi}(u)$ the {\em spacetime diagram} generated by $\Phi$ with initial condition $u$.
\end{definition}

For the cellular automata in this article, $\mathcal A=\F_p$. The configuration space $\F_{p}^\Z$ forms a group under componentwise addition; it is also an $\F_{p}$-vector space.

\begin{definition}\label{linear}
A cellular automaton $\Phi:\F_{p}^\Z\rightarrow \F_{p}^\Z$ is {\em linear} if $\Phi$ is an $\F_{p}$-linear map, i.e.\ 
 $(\Phi(u))_m = \alpha_{-\ell} u_{m-\ell} + \dots +\alpha_0 u_m + \dots + \alpha_r u_{m+r}$ for some nonnegative integers $\ell$ and $r$, called the {\em left} and {\em right radius} of $\Phi$. The {\em generating polynomial}~\cite{AvHPS:1997} of $\Phi$, denoted $\phi$, is the Laurent polynomial
\[ \phi(x) \colonequal \alpha_{-\ell}x^{\ell} + \dots +\alpha_0 + \dots +\alpha_r x^{-r}. \]
\end{definition}

We remark that our use of $\phi$ for the generating polynomial differs from usage in the literature of $\phi$ as $\Phi$'s local rule, which is the map 
$ (u_{m-\ell}, \dots, u_{m+r}) \mapsto \alpha_{-\ell} u_{m-\ell} + \dots + \alpha_r u_{m+r}$.

The generating polynomial has the property that $\phi(x) \sum_{m \in \Z} u_m x^m = \sum_{m \in \Z} (\Phi(u))_m x^m$.
We will identify sequences $(u_m)_{m \in \Z}$ with their generating function $f(x)=\sum_{m \in \Z}u_mx^m$.
Recall that $\F_{p}[x]$ and $\F_{p}\llbracket x\rrbracket$ are the rings of polynomials and power series in the variable $x$ with coefficients in $\F_{p}$ respectively. Let
$\F_{p}(x)$ and $\F_{p}((x))$ be their respective fields of fractions: $\F_{p}(x)$ is the field of rational functions and 
$\F_{p}((x))$ is that of formal Laurent series; elements of $\F_{p}((x))$ are expressions of the form $f(x)=\sum_{ m\geq m_{0}} u_{ m}x^{ m} $, where $u_{m}\in \F_{p}$ and $m_{0}\in \Z$.

\subsection{Cones}

A \emph{cone} is a subset of $\Z \times \Z$ of the form $\{\mathbf v_0 + s \mathbf v_1 + t \mathbf v_2 : s \geq 0, t \geq 0\}$ for some $\mathbf v_0, \mathbf v_1, \mathbf v_2 \in \Z \times \Z$ such that $\mathbf v_1$ and $\mathbf v_2$ are linearly independent.
The \emph{cone generated by $\mathbf v_1$ and $\mathbf v_2$} is the cone $\{s \mathbf v_1 + t \mathbf v_2 : s \geq 0, t \geq 0\}$.

If a cellular automaton is begun from an initial condition $u$ satisfying
$u_m = 0$ for all $m \leq -1$, then the spacetime diagram $\ST_\Phi(u)$ is supported on the cone generated by $(1, 0)$ and $(-r, 1)$.
For example, see Figure~\ref{0-TM}.
If $r \geq 1$ then this cone contains points with negative entries, but we would still like to represent $\ST_\Phi(u)$ as a formal power series in some ring.
We follow the geometric exposition given by Aparicio Monforte and Kauers~\cite{Aparicio--Kauers}.

\begin{figure}
	\center{\includegraphics[scale=.75]{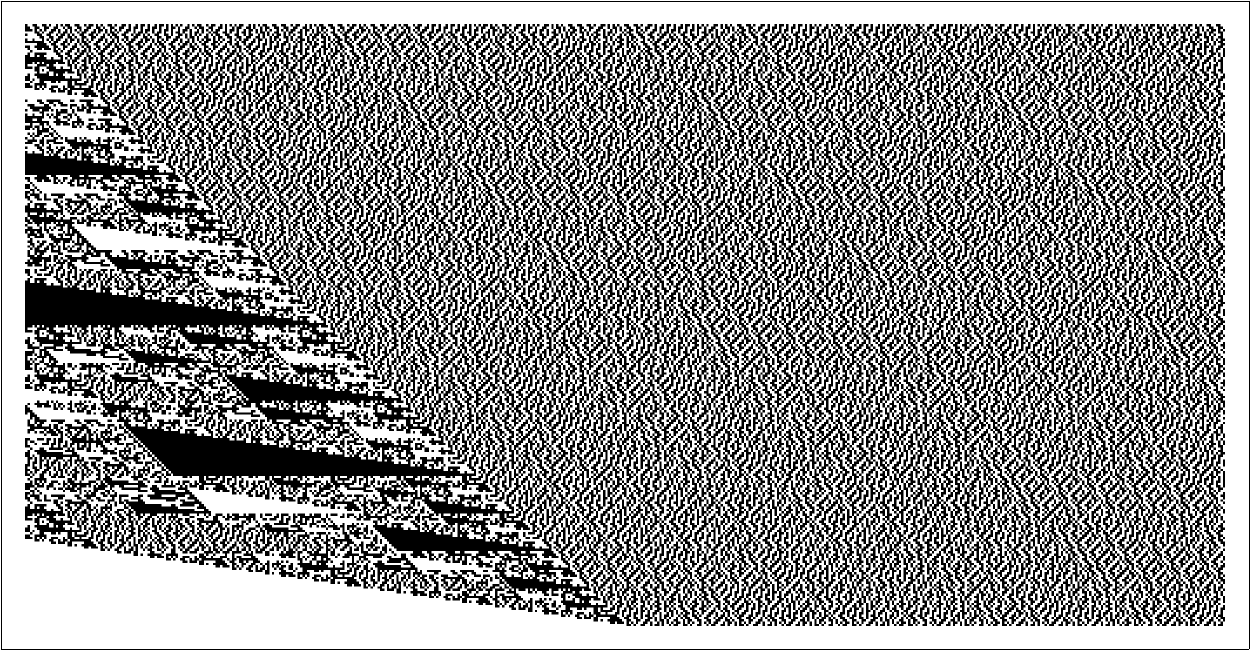}}
	\caption{Spacetime diagram $\ST_{\Phi}(u)$ for a cellular automaton with generating polynomial $\phi(x) = x^{-1} + x^{-3} + x^{-7} \in \F_2[x]$.
	The dimensions are $511 \times 256$, and time goes up the page.
	The right half $(u_m)_{m \geq 0}$ of the initial condition is the \emph{Thue--Morse sequence} (the fixed point beginning with $0$ of $0 \to 01, 1 \to 10$), and $u_m = 0$ for all $m \leq -1$.
	By Theorem~\ref{cone STD is algebraic}, this spacetime diagram, restricted to the cone generated by the vectors $(1,0)$ and $(-7,1)$, has an algebraic generating function.}
	\label{0-TM}
\end{figure}

By definition, a cone $\mathcal C$ is \emph{line-free}, that is, for every ${\mathbf n}\in \mathcal C \setminus \{ (0,0) \}$, we have $-{\mathbf n}\not\in \mathcal C$.
This places us within the scope of \cite{Aparicio--Kauers}.

For each cone $\mathcal C$, let $\F_{p, \mathcal C}\llbracket x,y\rrbracket$ be the set of all formal power series in $x$ and $y$, with coefficients in $\F_p$, whose support is in $\mathcal C$. Then (ordinary) multiplication of two elements in $\F_{p, \mathcal C}\llbracket x,y\rrbracket$ is well defined, and the product belongs to $\F_{p, \mathcal C}\llbracket x,y\rrbracket$; in fact
 $\F_{p, \mathcal C}\llbracket x,y\rrbracket$ is an integral domain~\cite[Theorems~10 and 11]{Aparicio--Kauers}.

Let $\preceq$ be the reverse lexicographic order on $\Z \times \Z$, i.e.\ $(m_1,n_1)\preceq(m_2,n_2)$ if $n_1< n_2$ or if $n_1= n_2$ and $m_1\leq m_2$. 
A cone $\mathcal C$ is {\em compatible} with $\preceq$ if $(0,0)\preceq (m, n)$ for all $(m, n) \in \mathcal C$.
Every cone contained in the set $\{(m,n):n>0 \}\cup \{(m,0): m\geq 0 \}$ is compatible with $\preceq$.
Let
\[
	\F_{p, \preceq}\llbracket x,y\rrbracket \colonequal\bigcup_\text{$\mathcal C$ compatible with $\preceq$} \F_{p, \mathcal C}\llbracket x,y\rrbracket .
\]
Then $\F_{p, \preceq}\llbracket x,y\rrbracket$ is a ring contained in the field $\bigcup_{(m,n) \in \Z \times \Z} x^{m}y^{n} \F_{p, \preceq}\llbracket x,y\rrbracket$
~\cite[Theorem~15]{Aparicio--Kauers}.
This field also contains the field $\F_{p}(x,y)$ of rational functions.
Researchers working with automatic sequences have previously worked with $\F_{p, \preceq}\llbracket x,y\rrbracket$~\cite{Adamczewski--Bell, Allouche--Deshouillers--Kamae--Koyanagi}.

\subsection{Automatic initial conditions}\label{automaticity}

Next we define automatic sequences, which we will use as initial conditions for spacetime diagrams.

\begin{definition}\label{DFAO}
A {\em deterministic finite automaton with output} (DFAO) is a $6$-tuple $(\mathcal S, \Sigma, \delta, s_0, \mathcal A, \omega)$, where $\mathcal S$ is a finite set (of states), $s_0 \in \mathcal S$ (the initial state), $\Sigma$ is a finite alphabet (the input alphabet), $\mathcal A$ is a finite alphabet (the output alphabet),
$\omega:\mathcal S\rightarrow \mathcal A$ (the output function), and $\delta:\mathcal S\times \Sigma \rightarrow \mathcal S$ (the transition function).
\end{definition}

In this article, our output alphabet is $\mathcal A= \F_p$.

The function $\delta$ extends in a natural way to the domain $\mathcal S \times \Sigma^*$, where $\Sigma^*$ is the set of all finite words on the alphabet $\Sigma$.
Namely, define $\delta(s, m_\ell \cdots m_1 m_0) \colonequal \delta(\delta(s, m_0), m_\ell \cdots m_1)$ recursively. If $\Sigma=\{0, \dots, p-1 \}$,
this allows us to feed the standard base-$p$ representation $m_\ell \cdots m_1 m_0$ of an integer $m$ into an automaton, beginning with the least significant digit.
(Recall that the standard base-$p$ representation of $0$ is the empty word.)
All automata in this article process integers by reading their least significant digit first.

A sequence $(u_m)_{m \geq 0}$ of elements in $\F_p$ is {\em $p$-automatic} if there is a DFAO $(\mathcal S, \{0, \dots, p - 1\}, \delta, s_0, \F_p, \omega)$ such that $u_m = \omega(\delta (s_0, m_\ell \cdots m_1 m_0))$ for all $m \geq 0$, where $m_\ell \cdots m_1 m_0$ is the standard base-$p$ representation of $m$.

Similarly, we say that a sequence $(U_{m,n})_{(m, n) \in \N \times \N}$ is {\em $[p, p]$-automatic} if there is a DFAO $(\mathcal S, \{0, \dots, p - 1\}^2, \delta, s_0, \F_p, \omega)$ such that
\[
	U_{m, n} = \omega(\delta(s_0, (m_\ell, n_\ell) \cdots (m_1, n_1) (m_0, n_0)))
\]
for all $(m, n) \in \N \times \N$, where $m_\ell \cdots m_1 m_0$ is a base-$p$ representation of $m$ and $n_\ell \cdots n_1 n_0$ is a base-$p$ representation of $n$. Here, if $m$ and $n$ have standard base-$p$ representations of different lengths, then we pad, on the left, the shorter representation with leading zeros.

As defined, $p$-automatic sequences are one-sided.
To specify a bi-infinite sequence, we use base $-p$.
Every integer has a unique representation in base $-p$ with the digit set $\{0, 1, \dots, p - 1\}$~\cite[Theorem~3.7.2]{Allouche--Shallit:2003}.
For example, $10$ is written in base $-2$ as
\begin{align*}
	10
	&= 16 - 8 + 4 - 2 + 0 \\
	&= 1 \cdot (-2)^4 + 1 \cdot (-2)^3 + 1 \cdot (-2)^2 + 1 \cdot (-2)^1 + 0 \cdot (-2)^0,
\intertext{so its base-$(-2)$ representation is $11110$, and }
	-9
	&= - 8 + 0 - 2 + 1 \\
	&= 1 \cdot (-2)^3 + 0 \cdot (-2)^2 + 1 \cdot (-2)^1 + 1 \cdot (-2)^0,
\end{align*}
so the base-$(-2)$ representation of $-9$ is $1011$. 
We say that a sequence $(u_m)_{m \in \Z}$ is \emph{$(-p)$-automatic} if there is a DFAO $(\mathcal S, \{0, \dots, p - 1\}, \delta, s_0, \F_p, \omega)$ such that $u_m = \omega(\delta (s_0, m_\ell \cdots m_1 m_0))$ for all $m \in \Z$, where $m_\ell \cdots m_1 m_0$ is the standard base-$(-p)$ representation of $m$.
A sequence $(u_m)_{m \in \Z}$ is $(-p)$-automatic if and only if the sequences $(u_m)_{m \geq 0}$ and $(u_{-m})_{m \geq 0}$ are $p$-automatic~\cite[Theorem~5.3.2]{Allouche--Shallit:2003}. 

In this article, we use $(-p)$-automatic sequences in $\F_p^\Z$ as initial conditions for cellular automata.
For example, the spacetime diagram in Figure~\ref{bi-infinite STD} is of a linear cellular automaton begun from a $(-2)$-automatic initial condition.

\begin{figure}
	\center{\includegraphics[scale=.75]{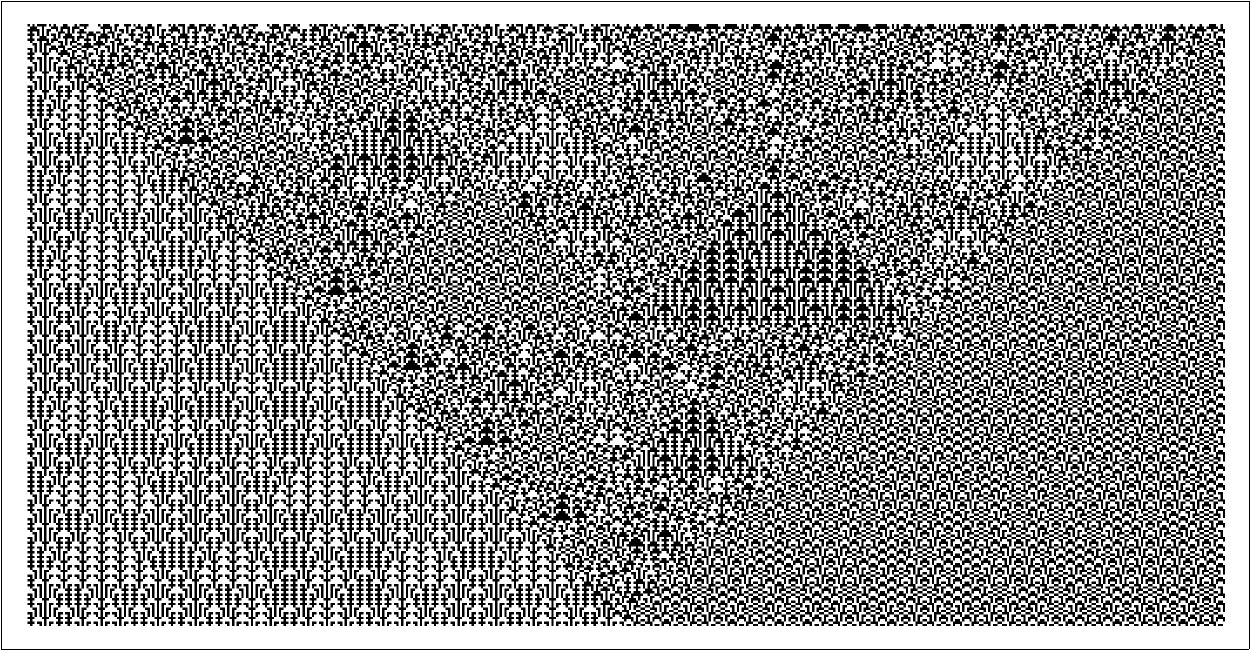}}
	\caption{Spacetime diagram for a cellular automaton with generating polynomial $\phi(x) = x + 1 + x^{-1} \in \F_2[x]$.
	The dimensions are $511 \times 256$.
	The right half $(u_m)_{m \geq 0}$ of the initial condition is the Thue--Morse sequence, and the left half $(u_{-m})_{m \geq 0}$ is the \emph{Toeplitz sequence} (the fixed point of $0 \to 01, 1 \to 00$).
	By Corollary~\ref{STD is automatic}, this spacetime diagram is $[-2, 2]$-automatic.}
	\label{bi-infinite STD}
\end{figure}

\section{Algebraicity and automaticity of spacetime diagrams}\label{Automaticity and algebraicity of spacetime diagrams}

In this section we show that a spacetime diagram obtained by evolving a linear cellular automaton from a $(-p)$-automatic initial condition $u$ is automatic in several senses. There is a natural notion of the $[p,p]$-kernel of a two-dimensional configuration extending the usual definition. First, if we consider bi-infinite initial conditions that satisfy $u_m = 0$ for all $m \leq -1$, we show in Theorem~\ref{cone STD is algebraic} that the generating functions of these cone-indexed configurations are algebraic and that they have finite $[p,p]$-kernels. Then in Section~\ref{automaticity by shearing} we show that the shear of an algebraic cone-indexed configuration is $[p,p]$-automatic.
Finally, in Section~\ref{Automaticity in base [-p, p]} we study the $[-p,p]$-automaticity of spacetime diagrams, where the coordinates $(m, n)$ are processed by reading $m$ in base $-p$. 
Specifically, we prove in Corollary~\ref{STD is automatic} that a spacetime diagram obtained by evolving a linear cellular automaton from a general $(-p)$-automatic initial condition is $[-p, p]$-automatic.

\subsection{Algebraicity and finiteness of the $[p, p]$-kernel}\label{algebraicity finite kernel}

Define the \emph{$[p, p]$-kernel} of $U = (U_{m, n})_{(m, n) \in \Z \times \N}$ to be the set
\[
	\left\{\left(U_{p^e m + i, p^e n + j}\right)_{(m, n) \in \Z \times \N} : e \geq 0, \, 0 \leq i \leq p^e - 1, \, 0 \leq j \leq p^e - 1\right\}.
\]
The $[p, p]$-kernel of a cone-indexed sequence $(U_{m, n})_{(m, n) \in \mathcal C}$ is defined by extending $U_{m, n} = 0$ for all $(m, n) \in \left(\Z\times \N \right)\setminus \mathcal C$.

Given $i, j \in \{0, 1, \dots, p - 1\}$, the {\em Cartier operator} 
 $\Lambda_{i,j}: \F_{p, \preceq}\llbracket x,y\rrbracket\rightarrow \F_{p, \preceq}\llbracket x,y\rrbracket$ is defined as
\[
	\Lambda_{i,j} \left(\sum_{(m,n)\in \mathcal C} U_{m,n} x^m y^n \right) \colonequal \sum_{(m,n): (mp+i,np+j)\in \mathcal C} U_{mp+i,np+j} x^{m}y^{n}.
\]

Let $\mathcal C$ be a cone. The {\em $[p, p]$-kernel} of a power series $F(x,y)=\sum_{(m,n)\in\mathcal C} U_{m,n}x^m y^n \in\F_{p, \mathcal C}\llbracket x,y\rrbracket$ is the set
\[
	\left\{ \Lambda_{i_\ell, j_\ell}\cdots \Lambda_{i_0,j_0} \left(F(x,y) \right) : \text{$\ell \geq 0$ and 
$0\leq i_k, j_k \leq p-1$ for $0 \leq k \leq \ell$} \right\}.
\]

If the sequence $(U_{m,n})_{(m,n)\in \mathcal C}$ is indexed by a cone, then its $[p, p]$-kernel is the set of all sequences
	$ (V_{m,n})_{(m,n)\in \mathcal C^*}$ where $ \sum_{(m,n)\in\mathcal C^*} V_{m,n}x^m y^n$ belongs to the $[p, p]$-kernel of
	$\sum_{(m,n)\in\mathcal C} U_{m,n}x^m y^n$. We show in Lemma~\ref{new cones are compatible} that such $\mathcal C^*$ are compatible with $\preceq$.

We can define analogously the one-dimensional Cartier operator $\Lambda_i:\F_p\llbracket x \rrbracket \to \F_p\llbracket x \rrbracket$ and also the $p$-kernel of a one-dimensional power series.
Eilenberg's theorem~\cite[Theorem~6.6.2]{Allouche--Shallit:2003} states that a sequence $(u_m)_{m\geq 0}$ is $p$-automatic precisely when its $p$-kernel is finite; the same is true for a $[p, p]$-automatic sequence $(U_{m, n})_{(m, n) \in \N \times \N}$~\cite[Theorem~14.4.1]{Allouche--Shallit:2003}. For a recent extension of Eilenberg's theorem to automatic sequences based on some alternative numeration systems, see \cite{Massuir-Peltomaki-Rigo}.

A power series $f(x)\in \F_p\llbracket x \rrbracket$ is {\em algebraic over $\F_{p}(x)$}
if there exists a nonzero polynomial $P(x, z) \in \F_{p}[x,z]$ such that $P(x, f(x)) = 0$.
Similarly, the cone-indexed series $f(x, y)\in \F_{p, \preceq}\llbracket x,y\rrbracket$ is {\em algebraic over $\F_{p}(x, y)$}
if there exists a nonzero polynomial $P(x, y, z) \in \F_{p}[x,y,z]$ such that $P(x, y, f(x)) = 0$.
We recall Christol's theorem for one-dimensional power series~\cite{Christol-1979, CKMR:1980}, generalised to two-dimensional power series by Salon~\cite{Salon-1987}. 

\begin{theorem}\label{Christol}
\leavevmode
\begin{enumerate}
\item
A sequence $(u_m)_{m \geq 0}$ of elements in $\F_p$ is $p$-automatic if and only if $\sum_{m \geq 0} u_m x^m$ is algebraic over $\F_{p}(x)$.
\item
A sequence of elements $(U_{m,n})_{(m, n) \in \N \times \N}$ in $\F_p$ is $[p, p]$-automatic if and only if $\sum_{(m, n) \in \N \times \N} U_{m,n} x^m y^n$ is algebraic over $\F_{p}(x,y)$.
\end{enumerate}
\end{theorem}

We refer to \cite[Theorems~12.2.5 and 14.4.1]{Allouche--Shallit:2003} for the proof of Theorem~\ref{Christol}, where it is shown that the algebraicity of a power series over a finite field is equivalent to the automaticity of its sequence of coefficients, which is equivalent to the finiteness of its $p$- or $[p, p]$-kernel. 
In related work, Allouche, Deshouillers, Kamae, and Koyanagi~\cite[Theorem~6]{Allouche--Deshouillers--Kamae--Koyanagi} show that the coefficients of an algebraic power series in $F_p((x))\llbracket y \rrbracket$ is $p$-automatic.

In the next lemma we show that the image of $\F_{p, \preceq}\llbracket x,y\rrbracket$ under $\Lambda_{i,j}$ is indeed contained in $\F_{p, \preceq}\llbracket x,y\rrbracket$.
We show more: although elements of the $[p, p]$-kernel of $F(x, y) \in \F_{p, \mathcal C}\llbracket x,y\rrbracket$ do not necessarily belong to $\F_{p, \mathcal C}\llbracket x,y\rrbracket$, their indexing sets are one of a finite set of translates of $\mathcal C$. 

\begin{lemma}\label{new cones are compatible}
Let $r \geq 0$ be an integer, let $\mathcal C$ be the cone generated by $(1, 0)$ and $(-r, 1)$, and let $F(x,y) 
\in \F_{p, \mathcal C}\llbracket x,y\rrbracket$.
Then every element of the $[p, p]$-kernel of $F(x,y)$ is supported on $\mathcal C - (t, 0)$ for some $0 \leq t \leq r$.
\end{lemma}

\begin{proof}
Let $0 \leq i \leq p - 1$, and $0 \leq j \leq p - 1$.
We abuse notation and define
$
	\Lambda_{i,j}(\mathcal C) \colonequal
	\left\{ \left(\tfrac{m - i}{p}, \tfrac{n - j}{p}\right) : (m,n) \in \mathcal C, \, m \equiv i \mod p, \, n \equiv j \mod p \right\}
$.
Let $0 \leq s \leq r$.
Then we claim that $\Lambda_{i, j}(\mathcal C - (s, 0)) = \mathcal C - (t, 0)$ for some $0 \leq t \leq r$.
The statement of the lemma follows from the claim.
Let $(m, n) \in \Z \times \Z$ be a point satisfying $n \geq 0$, $-m - r n \leq s$, $m \equiv i \mod p$, and $n \equiv j \mod p$.
Then $\Lambda_{i, j}$ maps $(m, n)$ to $\left(\frac{m - i}{p}, \frac{n - j}{p}\right)$, which satisfies $\frac{n - j}{p} \geq 0$ and
\[
	\textstyle{-\frac{m - i}{p} - r \cdot \frac{n - j}{p} \leq \frac{i + s + r j}{p} \leq \frac{(p - 1) + r + r (p - 1)}{p} = r + 1 - \frac{1}{p}}.
\]
Since $-\frac{m - i}{p} - r \cdot \frac{n - j}{p}$ is an integer, this implies $-\frac{m - i}{p} - r \cdot \frac{n - j}{p} \leq t \colonequal \left\lfloor\frac{i + s + r j}{p}\right\rfloor$ and $t \leq r$. 
\end{proof}

\begin{example}
If $p=2$ and $\mathcal C$ is generated by $(1,0)$ and $(-3,1)$, then $\Lambda_{0,0}$ and $\Lambda_{1,0}$ map $\mathcal C$ to itself.
The other Cartier operators map $\Lambda_{0,1}(\mathcal C) = \mathcal C - (1,0)$ and $\Lambda_{1,1}(\mathcal C) = \mathcal C - (2,0)$.
The cone $\mathcal C - (3,0)$ arises from $\Lambda_{1,1} \Lambda_{1,1}(\mathcal C) = \mathcal C - (3,0)$.
\end{example}

We now state Christol's theorem for $ \F_{p, \preceq}\llbracket x,y\rrbracket$. The case $r = 0$ is Salon's theorem (Part~(2) of Theorem~\ref{Christol}). We omit the proof, since it is a straightforward generalisation of the proofs in \cite[Theorems~12.2.5 and 14.4.1]{Allouche--Shallit:2003}.

\begin{theorem}\label{algebraic finite kernel}
Let $F(x,y) \in \F_{p, \preceq}\llbracket x,y\rrbracket$.
Then $F(x,y)$ is algebraic over $\F_{p}(x,y)$ if and only if $F(x,y)$ has a finite $[p, p]$-kernel.
\end{theorem}

Next we prove that a linear cellular automaton begun from a $p$-automatic initial condition produces an algebraic spacetime diagram.
A special case appears in Allouche et al.~\cite[Lemma~2]{AvHPS:1997}, when the initial condition is eventually $0$ in both directions. The proof in the general case is similar.

\begin{theorem}\label{cone STD is algebraic}
Let $\Phi:\F_p^{\Z}\rightarrow \F_p^{\Z}$ be a linear cellular automaton.
If $u \in \F_p^\Z$ is such that $(u_m)_{m \geq 0}$ is $p$-automatic and $u_m = 0$ for all $m \leq -1$, then the generating function of $\ST_\Phi(u)$ is algebraic and so has a finite $[p, p]$-kernel.
\end{theorem}

\begin{proof} 
Let the generating polynomial of $\Phi$ be
$ \phi(x) \colonequal \alpha_{-\ell}x^{\ell} + \dots +\alpha_0 + \dots +\alpha_r x^{-r}$.
Let $f_{u}(x)\in \F_p\llbracket x\rrbracket$ be the generating function of $ u$. 
The $n$-th row of $\ST_\Phi( u)$ is the sequence whose generating function is the Laurent series $\phi(x)^n f_{u}(x) $.
Let $\mathcal C$ be the cone generated by $(1,0)$ and $(-r,1)$.
Note that $U \colonequal \ST_\Phi(u)$ is identically $0$ on $\left(\Z\times \N \right)\setminus \mathcal C$, so its generating function satisfies $F_U(x,y)\in \F_{p,\mathcal C}\llbracket x,y\rrbracket\subseteq \F_{p,\preceq}\llbracket x,y\rrbracket$.
 Also,
\[
	F_U(x,y) = \sum_{n=0}^\infty \phi(x)^n f_{u}(x) y^n = \frac{f_{ u}(x)}{1-\phi(x) y}.
\]

Since $(u_m)_{m \geq 0}$ is $p$-automatic, Part~(1) of Theorem~\ref{Christol} guarantees the existence of a polynomial $P(x,z)\in \F_p[x,z]$ such that $P(x, f_u(x)) = 0$.
Let $Q(x, y, z) \colonequal P(x, (1 - \phi(x) y) z)$.
Then
\[
	Q(x, y, F_U(x, y))
	= P(x, (1 - \phi(x) y) F_U(x, y))
	= P(x, f_u(x))
	= 0,
\]
so $F_U(x,y)$ is algebraic. 
By Theorem~\ref{algebraic finite kernel}, $U=(U_{m,n})_{(m,n)\in \mathcal C}$ has a finite $[p, p]$-kernel.
\end{proof}

In Figure~\ref{0-TM} we have an illustration of a spacetime diagram satisfying the conditions of Theorem~\ref{cone STD is algebraic}.

Let $\mathcal C$ be the cone generated by $(1, 0)$ and $(-r, 1)$.
An interesting question is the following.
Given a polynomial equation $Q(x, y, z) = 0$ satisfied by $z = F(x,y) \in \F_{p, \mathcal C}\llbracket x,y\rrbracket$, is it decidable whether $F(x,y)$ is the generating function of $\ST_\Phi(u)$ for some linear cellular automaton $\Phi$?
The initial condition $u$ is determined by $F(x, 0)$, so it would suffice to obtain an upper bound on the left radius $\ell$.

\subsection{Automaticity by shearing}\label{automaticity by shearing}

If $r \geq 1$, then the cone generated by $(1, 0)$ and $(-r, 1)$ contains points $(m, n)$ where $m \leq -1$.
In this section, we feed these indices into an automaton by shearing the sequence so that it is supported on $\N \times \N$.

\begin{definition}\label{shear definition}
Let $\mathcal C$ be the cone generated by $(1,0)$ and $(-r,1)$, and let $s\geq 0$.
The {\em shear} of a sequence $(U_{m,n})_{(m,n) \in \mathcal C -(s,0)}$ is the sequence $(V_{m,n})_{(m,n) \in \N \times \N}$ defined by $V_{m,n}=U_{m - s - r n, n}$ for each $(m,n) \in \N \times \N$.
\end{definition}
 
The next lemma enables us to move between the $[p, p]$-kernel of the generating function $\sum_{(m,n) \in \mathcal C} U_{m,n} x^m y^n$ of a cone-indexed sequence and the generating function $\sum_{(m, n) \in \N \times \N}V_{m,n}x^m y^n$ of its shear.

\begin{lemma}\label{intertwining}
Let $F(x,y) \in \F_{p, \preceq}\llbracket x,y\rrbracket$.
Let $0 \leq i \leq p - 1$, $0 \leq j \leq p - 1$.
Then
\[
	\Lambda_{i, j}\left(x^\ell F(x, y)\right)
	=
	x^{-\left\lfloor\frac{i - \ell}{p}\right\rfloor} \Lambda_{(i - \ell) \bmod p, j}(F(x, y)).
\]
\end{lemma}

\begin{proof}
Let $\ell' = -\left\lfloor\frac{i - \ell}{p}\right\rfloor$.
Let $m, n \in \Z$.
We prove the result for the monomial $F(x, y) = x^m y^n$; the general result then follows from the linearity of $\Lambda_{i, j}$.
If $n \nequiv j \mod p$, then both sides are $0$.
If $n \equiv j \mod p$, we have
\begin{align*}
	\Lambda_{i, j}\left(x^\ell \cdot x^m y^n\right)
	&= \Lambda_{i, j}\left(x^{\ell + m} y^n\right) \\
	&= \begin{cases}
		x^\frac{\ell + m - i}{p} y^\frac{n - j}{p}	& \text{if $\ell + m \equiv i \mod p$} \\
		0							& \text{otherwise}
	\end{cases} \\
	&= \begin{cases}
		x^{\ell' + \frac{m - (i - \ell + p \ell')}{p}} y^\frac{n - j}{p}	& \text{if $m \equiv i - \ell + p \ell' \mod p$} \\
		0										& \text{otherwise}
	\end{cases} \\
	&= x^{\ell'} \Lambda_{i - \ell + p \ell', j}(x^m y^n) \\
	&= x^{\ell'} \Lambda_{(i - \ell) \bmod p, j}(x^m y^n).
	\qedhere
\end{align*}
\end{proof}

Note here that for each fixed $\ell$, the map $i \mapsto (i - \ell) \bmod p$ is a bijection.

\begin{example}
Let $p=3$, and let $F(x,y) \in \F_{3, \preceq}\llbracket x,y\rrbracket$.
For each $j$, we have
\begin{align*}
	\Lambda_{0,j}(x^{-1}F(x,y)) &= \Lambda_{1,j}(F(x,y)),\\
	\Lambda_{1,j}(x^{-1}F(x,y)) &= \Lambda_{2,j}(F(x,y)), \\
	\Lambda_{2,j}(x^{-1}F(x,y)) &= x^{-1}\Lambda_{0,j}(F(x,y)).
\end{align*}
\end{example}
 
We prove a version of Eilenberg's theorem for cone-indexed automatic sequences. We show there exists an explicit automaton representation of the shear of a cone-indexed $p$-automatic sequence using its $[p, p]$-kernel.

\begin{theorem}\label{Eilenberg shear}
Let $\mathcal C$ be generated by $(1,0)$ and $(-r,1)$ for some $r \geq 0$.
A $\mathcal C$-indexed sequence $(U_{m, n})_{(m, n) \in \mathcal C}$ of elements in $\F_p$ has a finite $[p, p]$-kernel if and only if its shear is $[p, p]$-automatic.
\end{theorem}

\begin{proof}
Let $(V_{m,n})_{(m, n) \in \N \times \N}$ be the shear of $(U_{m,n})_{(m,n) \in \mathcal C}$.
By \cite[Theorem~14.2.2]{Allouche--Shallit:2003}, $V$ is $[p, p]$-automatic if and only if its $[p, p]$-kernel is finite. Hence we show that $U$ has a finite $[p, p]$-kernel if and only if $V$ has a finite $[p, p]$-kernel.
 
By Lemma~\ref{new cones are compatible}, every element of the $[p, p]$-kernel of $U$ is supported on $\mathcal C - (s, 0)$ for some $0 \leq s \leq r$.
Let $W$ be an element of the $[p, p]$-kernel of $U$, supported on $\mathcal C - (s, 0)$.
Let $F(x, y) = \sum_{(m, n) \in \mathcal C - (s, 0)} W_{m, n} x^m y^n$.
Let $G_n(x,y) = x^{s + r n} \sum_{m \geq -s - r n} W_{m,n} x^m y^n$, so that
$
	G(x,y)
	= \sum_{n \geq 0} G_n(x,y)
$
is the generating function of the shear of $W$.
Similarly write $F_n(x,y)=\sum_{m\geq -s - r n} W_{m,n}x^m y^n$; then $G_n(x,y) = x^{s + r n} F_n(x,y)$.
Fix $n \equiv j \mod p$, and write $n = j + k p$ where $k\geq 0$.
By Lemma~\ref{intertwining}, we have
\begin{align*}
	\Lambda_{i,j}(F_n(x,y))
	&= \Lambda_{i,j}\left(x^{-s - r n} G_n(x,y)\right) \\
	&= x^{-\left\lfloor\frac{i + s + r n}{p}\right\rfloor} \Lambda_{(i + s + r n) \bmod p, j}(G_n(x, y)) \\
	&= x^{-\left\lfloor\frac{i + s + r j}{p}\right\rfloor - r k} \Lambda_{(i + s + r j) \bmod p, j}(G_{j + k p}(x, y)).
\end{align*}
Summing over $k \geq 0$ gives
\[
	\Lambda_{i,j}(F(x,y))
	= x^{-\left\lfloor\frac{i + s + r j}{p}\right\rfloor} \sum_{k \geq 0} x^{-r k} \Lambda_{(i + s + r j) \bmod p, j}(G_{j + k p}(x, y)).
\]
Therefore the shear of $x^{\left\lfloor\frac{i + s + r j}{p}\right\rfloor} \Lambda_{i,j}(F(x,y))$ is
\begin{align*}
	\sum_{k \geq 0} \Lambda_{(i + s + r j) \bmod p, j}(G_{j + k p}(x, y))
	&= \Lambda_{(i + s + r j) \bmod p, j} \left(\sum_{n \geq 0} G_{n}(x, y)\right) \\
	&= \Lambda_{(i + s + r j) \bmod p, j}\left(G(x,y)\right).
\end{align*}
Inductively, suppose $G(x, y)$ is the generating function of an element of the kernel of $V$.
Then the shear of $x^{\left\lfloor\frac{i + s + r j}{p}\right\rfloor} \Lambda_{i,j}(F(x, y))$ is an element of the $[p, p]$-kernel of $V$.
Note that $\Lambda_{i,j}(F(x, y))$ is supported on $\mathcal C- \left(\left\lfloor\frac{i + s + r j}{p}\right\rfloor, 0 \right)$.

We set up a map $\kappa$ from the $[p, p]$-kernel of $U$ to the $[p, p]$-kernel of $V$.
Let $\kappa(U) = V$, and define $\kappa$ recursively as follows.
For each $W$ in the $[p, p]$-kernel of $U$, let $\kappa(\Lambda_{i, j}(W)) = \Lambda_{(i + s+ r j) \bmod p, j}(\kappa(W))$ where $W$ is supported on $\mathcal C - (s, 0)$.
Since the map $i \mapsto (i +s+rj) \bmod p$ is a bijection on $\F_p$, $\kappa $ maps $\{ \Lambda_{i,j}(W): 0\leq i,j\leq p-1\} $ surjectively onto $\{ \Lambda_{i,j}(\kappa(W)): 0\leq i,j\leq p-1\} $.
It follows inductively that $\kappa$ is a surjection from the $[p,p]$-kernel of $U$ to the $[p,p]$-kernel of $V$.

If the $[p, p]$-kernel of $U$ is finite, then the surjectivity of $\kappa$ implies that the $[p, p]$-kernel of $V$ is finite.
By Lemma~\ref{new cones are compatible}, $\kappa$ is at most $(r+1)$-to-one.
Therefore if the $[p, p]$-kernel of $V$ is finite then the $[p, p]$-kernel of $U$ has at most $r+1$ times as many elements and is also finite.
\end{proof}

We can now extend Theorem~\ref{cone STD is algebraic}.

\begin{theorem}\label{sheared STD is automatic}
Let $\Phi:\F_p^{\Z}\rightarrow \F_p^{\Z}$ be a linear cellular automaton.
If $u \in \F_p^\Z$ is such that $(u_m)_{m \geq 0}$ is $p$-automatic and $u_m = 0$ for all $m \leq -1$, then the shear of $\ST_\Phi(u)$ is $[p, p]$-automatic.
\end{theorem}

\begin{proof}
By Theorem~\ref{cone STD is algebraic}, $\ST_\Phi(u)$ has a finite $[p, p]$-kernel. By Theorem~\ref{Eilenberg shear}, we conclude that the shear of $\ST_\Phi(u)$ is $[p, p]$-automatic.
\end{proof}

\begin{figure}
	\center{$\vcenter{\hbox{\includegraphics[scale=.7]{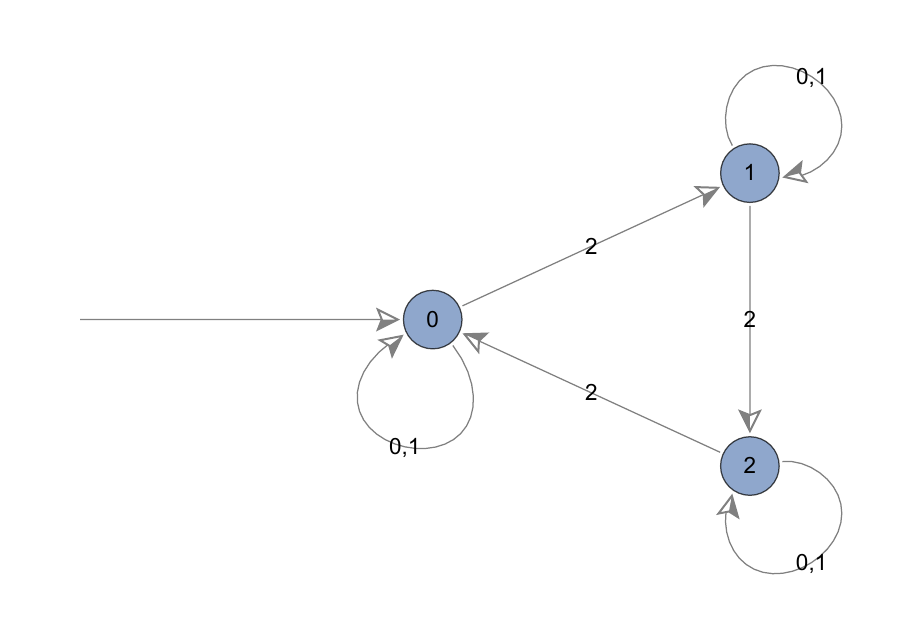}}} \qquad \vcenter{\hbox{\includegraphics[scale=.5]{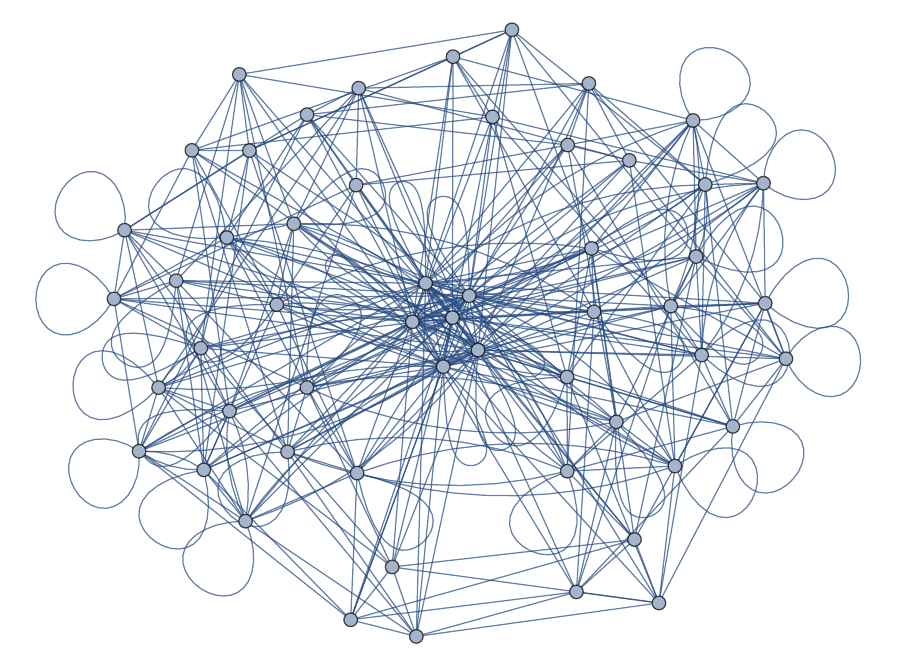}}}$}
	\caption{Automata from Example~\ref{p=3 example}. The automaton on the left generates the initial condition $u$, and the automaton on the right generates the spacetime diagram $\ST_\Phi(u)$.}
	\label{automata}
\end{figure}

\begin{example}\label{p=3 example}
Let $p = 3$, and let $\phi(x) = x + 1 \in \F_3[x]$.
Let $(u_m)_{m \geq 0}$ be the $3$-automatic sequence generated by the automaton on the left in Figure~\ref{automata}, whose first few terms are $001001112\cdots$.
The size of this automaton makes later computations feasible.
Let $u_m = 0$ for all $m \leq -1$; then the spacetime diagram $U = \ST_\Phi(u)$ is supported on $\N \times \N$.
See Figure~\ref{p=3 right part}.
We compute an automaton for the $[3, 3]$-automatic sequence $U|_{\N \times \N}$.
By Part~(1) of Theorem~\ref{Christol}, we can compute a polynomial $P(x, y)$ such that $P(x, f_u(x)) = 0$.
We compute
\begin{multline*}
	P(x, y) = x^{28} y
	+ 2 \left(x^{12} + x^{21} + x^{24} + x^{27} + x^{28} + x^{29}\right) y^3 \\
	+ \left(1 + 2 x^9 + x^{12} + x^{15} + x^{18} + x^{21} + 2 x^{24} + 2 x^{27} + x^{30}\right) y^9 \\
	+ 2 \left(1 + x^{27} + x^{54}\right) y^{27}.
\end{multline*}
Note that this is not the minimal polynomial for $f_u(x)$, but it is in a convenient form for the subsequent computation.
As in the proof of Theorem~\ref{cone STD is algebraic}, the generating function $F_U(x, y)$ of $U$ satisfies $P(x, (1 - \phi(x) y) F_U(x, y)) = 0$.
By Part~(2) of Theorem~\ref{Christol}, we can use this polynomial equation to compute an automaton for $U|_{\N \times \N}$.
The resulting automaton has $486$ states; minimizing produces an equivalent automaton with $54$ states.
This automaton is shown without labels or edge directions on the right in Figure~\ref{automata}.
These computations were performed with the \textit{Mathematica} package \textsc{IntegerSequences}~\cite{IntegerSequences}.
\end{example}

\begin{figure}
	\center{\includegraphics[scale=.75]{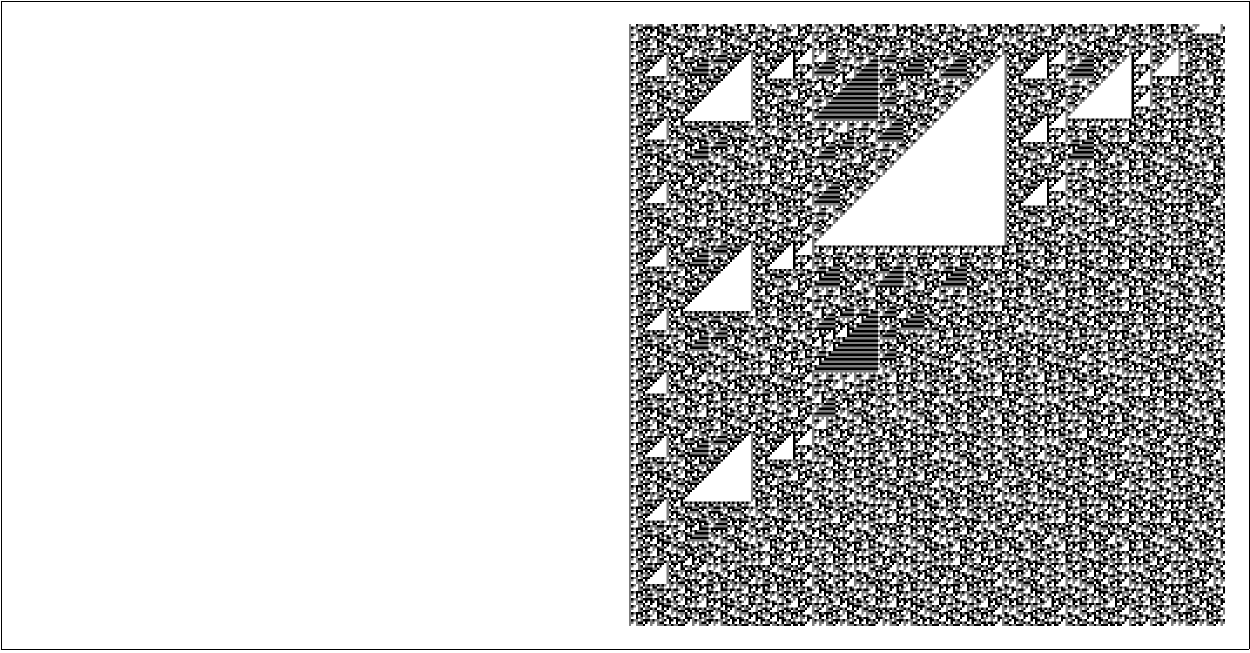}}
	\caption{Spacetime diagram for a cellular automaton with generating polynomial $\phi(x) = x + 1 \in \F_3[x]$.
	The initial condition is generated by the automaton in Example~\ref{p=3 example}.
	The line $n = m$ separates the diagram into two regions; the upper region contains arbitrarily large white patches, and the lower region does not.
	This is because the left half of the initial condition is identically $0$.
	The dimensions are $511 \times 256$.}
	\label{p=3 right part}
\end{figure}

\subsection{Automaticity in base $[-p, p]$}\label{Automaticity in base [-p, p]}

Instead of shearing, we may evaluate an automaton at negative integers by using base $-p$.
This approach gives a variant of Theorem~\ref{Eilenberg shear} and a notion of automaticity of $\ST_\Phi(u)$ for a general $(-p)$-automatic initial condition $u$.

\begin{definition}\label{[-p,p]}
A sequence $(U_{m, n})_{(m, n) \in \Z \times \N}$ is \emph{$[-p, p]$-automatic} if there is a DFAO $(\mathcal S, \{0, \dots, p - 1\}^2, \delta, s_0, \F_p, \omega)$ such that
\[
	U_{m, n} = \omega(\delta(s_0, (m_\ell, n_\ell) \cdots (m_1, n_1) (m_0, n_0)))
\]
for all $(m, n) \in \N \times \N$, where $m_\ell \cdots m_1 m_0$ is the standard base-$(-p)$ representation of $m$ and $n_\ell \cdots n_1 n_0$ is the standard base-$p$ representation of $n$, padded with zeros if necessary, as in Section~\ref{automaticity}.
\end{definition}

\begin{theorem}\label{STD Eilenberg}
A sequence $(U_{m, n})_{(m, n) \in \Z \times \N}$ has a finite $[p, p]$-kernel if and only if it is $[-p, p]$-automatic.
\end{theorem}

\begin{proof}
Define the $[-p, p]$-Cartier operator $\bar{\Lambda}_{i, j}$ by
\[
	\bar{\Lambda}_{i, j} \left((W_{m, n})_{(m, n) \in \Z \times \N}\right)
	\colonequal (W_{-p m + i, p n + j})_{(m, n) \in \Z \times \N}.
\]
Define the \emph{$[-p, p]$-kernel} of $U = (U_{m, n})_{(m, n) \in \Z \times \N}$ to be
the smallest set containing $U$ that is closed under $\bar{\Lambda}_{i, j}$ for all $i, j \in \{0, 1, \dots, p - 1\}$.
We show that the $[p, p]$-kernel of $U$ is finite if and only if the $[-p, p]$-kernel of $U$ is finite.

For a sequence $(W_{m, n})_{(m, n) \in \Z \times \N}$, define $\rho(W) \colonequal (W_{-m, n})_{(m, n) \in \Z \times \N}$ and $\sigma^{-1}(W) \colonequal (W_{m - 1, n})_{(m, n) \in \Z \times \N}$.
Let $K$ be the union, over all elements $W$ in the $[p, p]$-kernel of $U$, of the set
\[
	\left\{
		W, \,
		\rho(W), \,
		\sigma^{-1}(W), \,
		\rho(\sigma^{-1}(W))
	\right\}.
\]
We claim that the $[-p, p]$-kernel of $U$ is a subset of $K$.
One verifies that $\bar{\Lambda}_{i, j}(K) \subseteq K$:
\begin{align*}
	\bar{\Lambda}_{i, j}(W) &= \rho(\Lambda_{i, j}(W)) \\
	\bar{\Lambda}_{i, j}(\rho(W)) &= \begin{cases}
		\Lambda_{0, j}(W)				& \text{if $i = 0$} \\
		\sigma^{-1}(\Lambda_{p - i, j}(W))	& \text{if $i \neq 0$}
	\end{cases} \\
	\bar{\Lambda}_{i, j}(\sigma^{-1}(W)) &= \begin{cases}
		\rho(\sigma^{-1}(\Lambda_{p - 1, j}(W)))	& \text{if $i = 0$} \\
		\rho(\Lambda_{i - 1, j}(W))				& \text{if $i \neq 0$}
	\end{cases} \\
	\bar{\Lambda}_{i, j}(\rho(\sigma^{-1}(W))) &= \sigma^{-1}(\Lambda_{p - 1 - i, j}(W)).
\end{align*}
For example, if $i \neq 0$ we have
\begin{align*}
	\bar{\Lambda}_{i, j}(\rho(W))
	&= \bar{\Lambda}_{i, j}\left((W_{-m, n})_{(m, n) \in \Z \times \N}\right) \\
	&= (W_{-(-p m + i), p n + j})_{(m, n) \in \Z \times \N} \\
	&= (W_{p (m - 1) + p - i, p n + j})_{(m, n) \in \Z \times \N} \\
	&= \sigma^{-1}\left((W_{p m + p - i, p n + j})_{(m, n) \in \Z \times \N}\right) \\
	&= \sigma^{-1}(\Lambda_{p - i, j}(W));
\end{align*}
the other identities follow similarly.
Since $U \in K$, it follows that the $[-p, p]$-kernel of $U$ is a subset of $K$.
Therefore there are at most four times as many elements in the $[-p, p]$-kernel as in the $[p, p]$-kernel, so if the $[p, p]$-kernel is finite then the $[-p, p]$-kernel is also finite.

Similarly, we can emulate $\Lambda_{i, j}$ by taking the four states $W, \rho(W), \sigma(W), \sigma(\rho(W))$ for each element $W$ in the $[-p, p]$-kernel of $U$, where $\sigma(W) \colonequal (W_{m + 1, n})_{(m, n) \in \Z \times \N}$:
\begin{align*}
	\Lambda_{i, j}(W) &= \rho(\bar{\Lambda}_{i, j}(W)) \\
	\Lambda_{i, j}(\rho(W)) &= \begin{cases}
		\bar{\Lambda}_{0, j}(W)			& \text{if $i = 0$} \\
		\sigma(\bar{\Lambda}_{p - i, j}(W))	& \text{if $i \neq 0$}
	\end{cases} \\
	\Lambda_{i, j}(\sigma(W)) &= \begin{cases}
		\sigma(\rho(\bar{\Lambda}_{0, j}(W)))	& \text{if $i = p - 1$} \\
		\rho(\bar{\Lambda}_{i + 1, j}(W))		& \text{if $i \neq p - 1$}
	\end{cases} \\
	\Lambda_{i, j}(\sigma(\rho(W))) &= \sigma(\bar{\Lambda}_{p - 1 - i, j}(W)).
\end{align*}
It follows that there are at most four times as many elements in the $[p, p]$-kernel as in the $[-p, p]$-kernel, so if the $[-p, p]$-kernel is finite then the $[p, p]$-kernel is also finite.

Now we show that the $[-p, p]$-kernel of $U$ is finite if and only if $U$ is $[-p, p]$-automatic.
The proof is similar to the usual proof of Eilenberg's characterisation, as in \cite[Theorem 6.6.2]{Allouche--Shallit:2003}.
If the $[-p, p]$-kernel of $U$ is finite, then the automaton whose states are the elements of the $[-p, p]$-kernel and whose transitions are determined by the action of $\bar{\Lambda}_{i, j}$ is finite; moreover, this automaton outputs $U_{m, n}$ when fed the base-$[-p, p]$ representation of $(m, n)$.
Conversely, if there is such an automaton, then the $[-p, p]$-kernel is finite since it can be embedded into the set of states of the automaton.
\end{proof}

\begin{theorem}\label{cone STD is automatic}
Let $\Phi:\F_p^{\Z}\rightarrow \F_p^{\Z}$ be a linear cellular automaton.
If $u \in \F_p^\Z$ is such that $(u_m)_{m \geq 0}$ is $p$-automatic and $u_m = 0$ for all $m \leq -1$, then $\ST_\Phi(u)$ is $[-p, p]$-automatic.
\end{theorem}

\begin{proof}
By Theorem~\ref{cone STD is algebraic}, $\ST_\Phi(u)$ has a finite $[p, p]$-kernel.
By Theorem~\ref{STD Eilenberg}, $\ST_\Phi(u)$ is $[-p, p]$-automatic.
\end{proof}

\begin{corollary}\label{STD is automatic}
Let $\Phi:\F_p^{\Z}\rightarrow \F_p^{\Z}$ be a linear cellular automaton.
If $u \in \F_p^\Z$ is $(-p)$-automatic, then $\ST_\Phi(u)$ is $[-p, p]$-automatic.
\end{corollary}

\begin{proof}
Consider the two initial conditions $\cdots u_{-2} u_{-1} \cdot 0 0 \cdots$ and $\cdots 0 0 \cdot u_0 u_1 \cdots$.
By Theorem~\ref{cone STD is automatic}, $\ST_\Phi(\cdots 0 0 \cdot u_0 u_1 \cdots)$ is $[-p, p]$-automatic.
A straightforward modification of Theorem~\ref{cone STD is automatic} shows that $\ST_\Phi(\cdots u_{-2} u_{-1} \cdot 0 0 \cdots)$ is also $[-p, p]$-automatic.
Since $\Phi$ is linear, $\ST_\Phi(u)$ is the termwise sum of these two spacetime diagrams. The sum of two $[-p, p]$-automatic sequences is automatic; therefore $\ST_\Phi(u)$ is $[-p, p]$-automatic.
\end{proof}

\begin{figure}
	\center{\includegraphics[scale=.75]{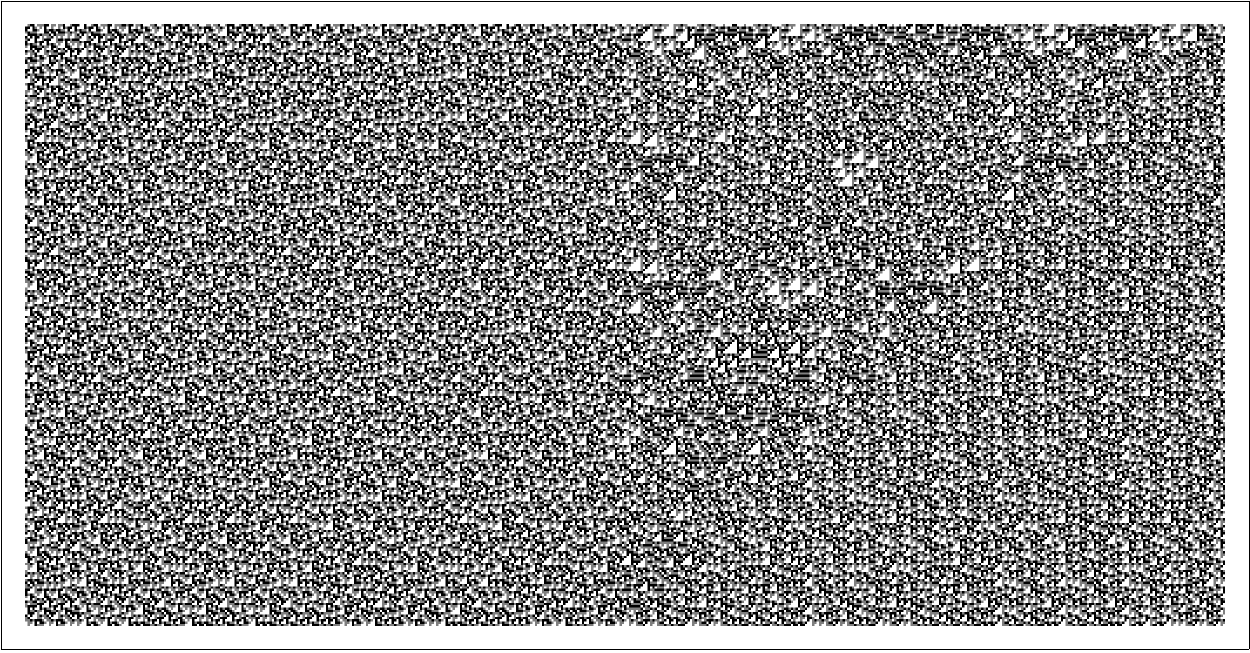}}
	\caption{Spacetime diagram for the linear cellular automaton with generating polynomial $\phi(x) = x + 1 \in \F_3[x]$ begun from the $3$-automatic initial condition described in Example~\ref{p=3 example part 2}.
	The dimensions are $511 \times 256$.}
	\label{p=3 both parts}
\end{figure}

\begin{example}\label{p=3 example part 2}
As in Example~\ref{p=3 example}, let $p = 3$, let $\phi(x) = x + 1 \in \F_3[x]$, and let $(u_m)_{m \geq 0}$ be $3$-automatic sequence generated by the automaton on the left in Figure~\ref{automata}.
We extend $(u_m)_{m \geq 0}$ to a $(-3)$-automatic sequence $(u_m)_{m \in \Z}$ by setting $u_m = u_{-m}$ for all $m \leq -1$.
The resulting spacetime diagram is shown in Figure~\ref{p=3 both parts}.
By Corollary~\ref{STD is automatic}, $\ST_\Phi(u)$ is $[-3, 3]$-automatic.

To compute an automaton for $\ST_\Phi(u)$, we start with the $54$-state automaton computed in Example~\ref{p=3 example} for the right half $(U_{m, n})_{(m, n) \in \N \times \N}$ of the spacetime diagram in Figure~\ref{p=3 right part}.
We convert this $[3, 3]$-automaton using Theorem~\ref{STD Eilenberg} to a $[-3, 3]$-automaton for the spacetime diagram $(U_{m, n})_{(m, n) \in \Z \times \N}$ in Figure~\ref{p=3 right part} whose left half is identically $0$; minimizing produces an automaton $\mathcal M$ with $204$ states.

We also need an automaton for the $\Z \times \N$-indexed spacetime diagram with initial condition $\cdots u_{-2} u_{-1} 0 0 0 \cdots$, shown in Figure~\ref{p=3 left part}.
The symmetry $x^{-1} \phi(x) = \phi(x^{-1})$ implies that a shear of this diagram is the left--right reflection $(U_{-m, n})_{(m, n) \in \Z \times \N}$ of the diagram in Figure~\ref{p=3 right part}.
Since $(U_{-m, n})_{(m, n) \in \Z \times \N}$ is an element of the $[-3, 3]$-kernel of $U$, we obtain an automaton for $(U_{-m, n})_{(m, n) \in \Z \times \N}$ simply by changing the initial state in $\mathcal M$ to be the state corresponding to this kernel sequence; hence $(U_{-m, n})_{(m, n) \in \Z \times \N}$ is generated by an automaton $\mathcal M'$ with $204$ states.
Shearing $(U_{-m, n})_{(m, n) \in \Z \times \N}$ produces $(U_{-m + n, n})_{(m, n) \in \Z \times \N}$, the spacetime diagram in Figure~\ref{p=3 left part}.
Using a variant of Theorem~\ref{Eilenberg shear} for the $[-p, p]$-kernel of a $\Z \times \N$-indexed sequence, we compute an automaton with $204$ states for this spacetime diagram.

Finally, since $u_0 = 0$, the product of the automata for $(U_{m, n})_{(m, n) \in \Z \times \N}$ and $(U_{-m + n, n})_{(m, n) \in \Z \times \N}$ is an automaton for the sum $\ST_\Phi(u)$ of the spacetime diagrams in Figures~\ref{p=3 right part} and \ref{p=3 left part}, which is the diagram in Figure~\ref{p=3 both parts}.
The product automaton has $204^2$ states, but minimizing reduces this to $1908$ states.
\end{example}

\begin{figure}
	\center{\includegraphics[scale=.75]{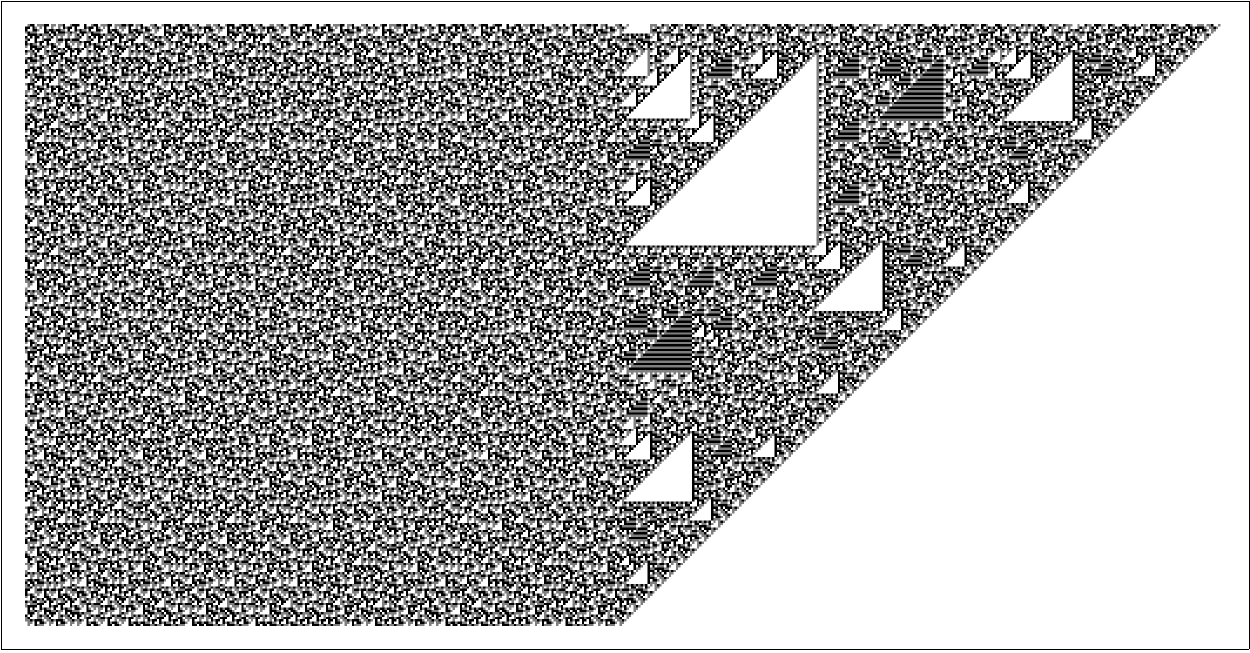}}
	\caption{Spacetime diagram whose sum with the diagram in Figure~\ref{p=3 right part} is the diagram in Figure~\ref{p=3 both parts}.}
	\label{p=3 left part}
\end{figure}

\section{Automaticity of $\Z\times \Z$-indexed spacetime diagrams}\label{Automaticity of Z x Z diagrams}

In Corollary~\ref{STD is automatic}, we showed that if $u$ is $(-p)$-automatic then the $\Z\times \N$-configuration $\ST_\Phi(u)$ is $[-p,p]$-automatic. 
Our aim in this section is to extend Corollary~\ref{STD is automatic} to $\Z \times \Z$-configurations. We remark that the results of this section can be further extended to statements about two-dimensional linear recurrences with constant coefficients. We also note that Bousquet-M\'{e}lou and Petko\v{v}sek~\cite{Bousquet-Melou--Petkovsek-2000} prove similar results, with different proofs, for linear recurrences on $\N\times \N$ over fields of characteristic $0$.

\begin{definition}
If $U\in \F_p^{\Z \times\Z }$ satisfies $\Phi(U|_{\Z\times \{n\}}) = U|_{\Z\times \{n + 1\}}$ for each $n\in \Z$, we call $U$ a {\em spacetime diagram} for $\Phi$.
\end{definition}

Note that if $\Phi: \F_p^\Z \rightarrow \F_p^\Z$ is a linear cellular automaton with left and right radii $\ell$ and $r$ respectively, then it is surjective, and every sequence in $\F_p^\Z$ has $p^{\ell+r}$ preimages.
Hence if $\ell+r\geq 1$ there are infinitely many $\Z\times \Z$-indexed spacetime diagrams $U$ such that $U|_{\Z \times \{0\}} = u$.

Let $\Phi$ have generating polynomial $\phi(x) = \alpha_{-\ell}x^{\ell} + \dots +\alpha_0 + \dots +\alpha_r x^{-r}$.
A configuration $U=(U_{m,n})_{(m,n)\in \Z \times \Z}$ is a spacetime diagram for $\Phi$ if and only if
\[
	(1 - \phi(x) y) \sum_{(m,n) \in \Z \times \Z} U_{m,n} x^m y^n = 0.
\]
In the following lemma we identify which initial conditions determine a spacetime diagram for $\Phi$.

\begin{lemma}\label{initial conditions for LCA}
Let $\Phi:\F_p^\Z \rightarrow \F_p^\Z$ be a linear cellular automaton with generating polynomial $\phi(x) = \alpha_{-\ell}x^{\ell} + \dots +\alpha_0 + \dots +\alpha_r x^{-r}$. Let 
\[
	\I = (\Z\times \{ 0 \}) \cup \bigcup_{i=0}^{\ell+r-1}(\{ i \} \times -\N).
\]
Then every $U \in \F_p^\I$ can be uniquely extended to a spacetime diagram $U \in \F_p^{\Z \times \Z}$ for $\Phi$.
\end{lemma}

\begin{proof}
Note that $U|_{\Z\times \{0\}}$ uniquely determines a $\Z \times \N$-indexed spacetime diagram for $\Phi$.
Next we observe that $U|_{( \Z \times \{0\} ) \cup \{(0,-1), \dots, (\ell+r-1,-1)\}}$ determines $U|_{\Z \times \{ -1\}}$ for $\Phi$.
 For, given a word $w\in \F_p^{\ell+r}$, there is a unique sequence $v\in \F_p^\Z$ such that $v_0 \cdots v_{\ell+r-1}=w$ and $\Phi(v)= U|_{\Z \times \{0\}}$. 
Similarly, $U|_{( \Z \times \{-n\} ) \cup \{(0,-n-1), \dots, (\ell+r-1,-n-1)\}}$ determines $U|_{\Z \times \{ -n-1\}}$.
 We can repeat this, determining one row at a time, once we have specified a word of length $\ell +r $ in that row.
\end{proof}

\begin{figure}
	\center{\includegraphics[scale=.75]{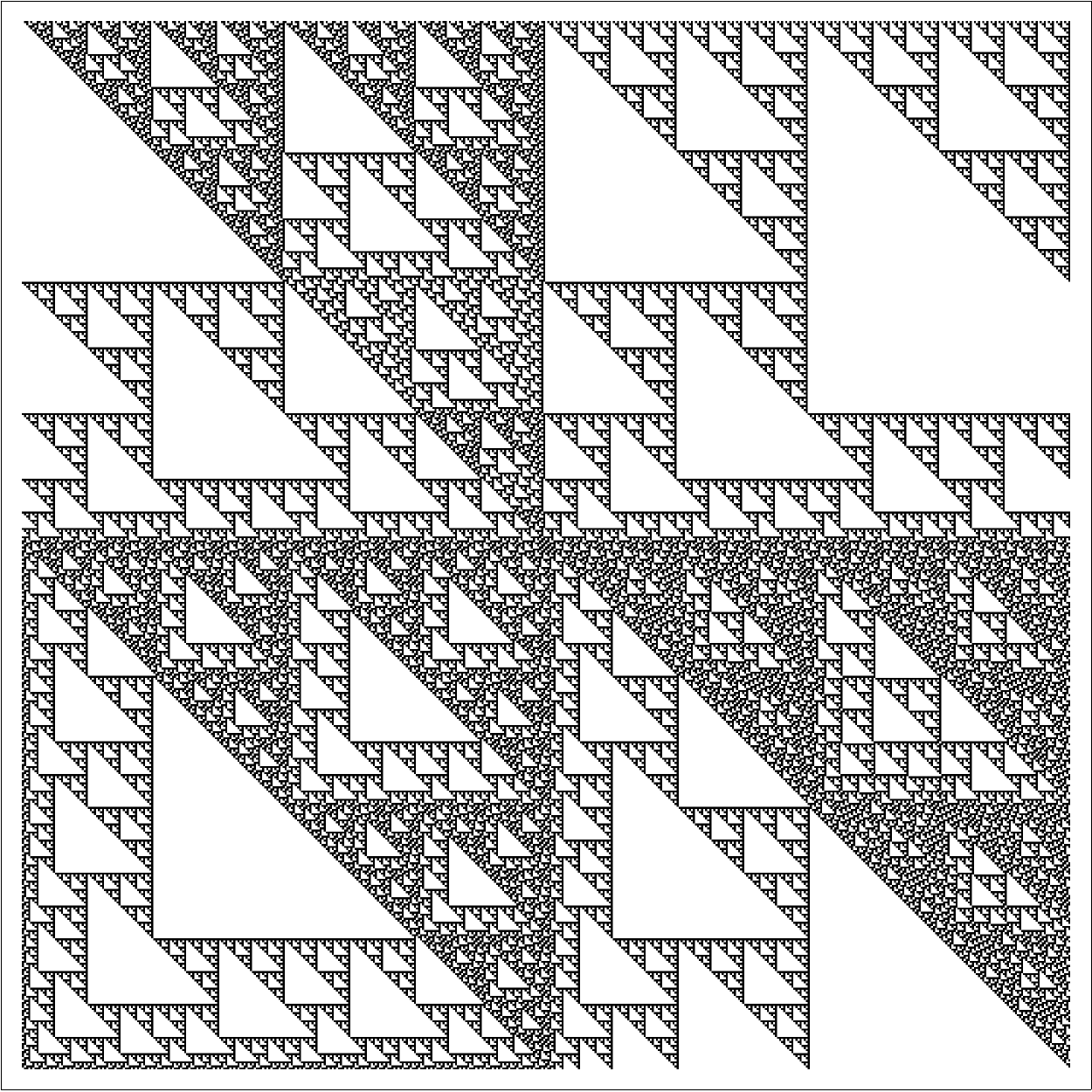}}
	\caption{A $\Z \times \Z$-indexed spacetime diagram for the Ledrappier cellular automaton.
	The initial conditions are $U_{m, 0} = T(m)$ for $m \geq 0$, $U_{m, 0} = T(-m)$ for $m \leq -1$, and $U_{0, n} = T(-n)$ for $n \leq -1$, where $T(m)_{m \geq 0}$ is the Thue--Morse sequence.
	The dimensions are $511 \times 511$.}
	\label{Z x Z}
\end{figure}

\begin{example}
Consider the {\em Ledrappier} cellular automaton $\Phi$, whose generating polynomial is $\phi(x)=1+x^{-1}$.
By Lemma~\ref{initial conditions for LCA}, $U$ is determined by its values on $(\Z \times \{0\})\cup (\{0\} \times -\N)$.
See Figure~\ref{Z x Z} for an example of a spacetime diagram for $\Phi$.
\end{example}

Definition~\ref{[-p,p]} naturally generalises to $[p,q]$-automaticity for any integers $p, q$ with $|p| \geq 2$ and $|q| \geq 2$.
Therefore we may consider $[-p,-p]$-automaticity.
One can also define $[p, p]$-automaticity for any of the four quadrants $(\pm\N) \times (\pm\N)$.

\begin{proposition}\label{-p iff p}
A sequence $U \in \F_p^{\Z\times\Z}$ is $[-p, -p]$-automatic if and only if each of $U|_{( \pm\N) \times (\pm\N)}$ is $[p, p]$-automatic.
\end{proposition}

The proof of Proposition~\ref{-p iff p} follows the same lines as that of \cite[Theorem~5.3.2]{Allouche--Shallit:2003}.

\begin{theorem}\label{2D-STDs_are_automatic}
Let $\Phi:\F_p^\Z \rightarrow \F_p^\Z$ be a linear cellular automaton with left and right radii $\ell$ and $r$. Let $U \in \F_p^{\Z \times \Z}$ be a spacetime diagram for $\Phi$. 
If $U|_{ \{i\}\times -\N }$ is $p$-automatic for each $i$ in the interval $-\ell \leq i\leq r-1$ and $U|_{\Z\times \{0\}}$ is $(-p)$-automatic, 
 then
 $U$ is $[-p, -p]$-automatic.
\end{theorem}

\begin{proof} 
By Lemma~\ref{initial conditions for LCA}, $U$ is uniquely determined by its values on $(\Z\times \{0\}) \cup \bigcup_{i=-\ell}^r (\{i\}\times -\N)$.
By Proposition~\ref{-p iff p}
it is sufficient to show that each of the four quadrants $U|_{(\pm\N) \times (\pm \N)}$ is $[p,p]$-automatic. 

By Corollary~\ref{STD is automatic}, $U|_{\Z \times \N}$ is $[-p,p]$-automatic. By Theorem~\ref{STD Eilenberg}, $U|_{\Z \times \N}$ has a finite $[p,p]$-kernel. Thus each of $U|_{\pm \N \times \N}$ has a finite $[p,p]$-kernel. By Theorem~\ref{Eilenberg shear} with $r=0$, each of $U|_{\pm \N \times \N}$ is $[p,p]$-automatic.

We show that $U|_{\N \times -\N}$ is $[p,p]$-automatic; the automaticity of $U|_{-\N \times -\N}$ follows by a similar argument.
Let $\phi(x) = \alpha_{-\ell}x^{\ell} + \dots +\alpha_0 + \dots +\alpha_r x^{-r}$ be the generating polynomial of $\Phi$.
For $S\subseteq \Z\times \Z$, let $F|_S$ denote the generating function of $U|_S$.
Since $U$ is a spacetime diagram for $\Phi$, we have $U_{m,n+1}-\sum_{i=-\ell}^r \alpha_i U_{m+i,n} = 0$ for each $(m,n)\in \Z\times \Z$.
Multiplying by $x^{m}y^{n+1}$ and summing over $m\geq 0$ and $n\leq -1$ gives
\begin{align*}
	0
	&= \sum_{\substack{m\geq 0 \\ n\leq -1} }U_{m,n+1}x^m y^{n+1} - \sum_{\substack{m\geq 0 \\ n\leq -1}}\sum_{i=-\ell}^r \alpha_i U_{m+i, n} x^{m}y^{n+1}\\ 
	&= F|_{\N \times -\N} - \sum_{i=-\ell}^r\alpha_i x^{-i}y\left(\sum_{\substack{m\geq 0 \\ n\leq -1}}U_{m+i,n} x^{m+i}y^{n}\right) \\
 	&= F|_{\N \times -\N} - \sum_{i=-\ell}^{-1}\alpha_i x^{-i}y \left(
		\sum_{k=-\ell }^{i} F|_{\{k\}\times -\N}
		+ F|_{\N \times -\N}
		- F|_{\N \times \{0\}}
		- P_i(x)
	\right) \\
	&\phantom{00}
	- \alpha_0y \left(
		F|_{\N \times -\N}
		- F|_{\N \times \{0\}}
	\right)
	- \sum_{i=1}^{r - 1}\alpha_i x^{-i}y \left(
		F|_{\N \times -\N}
		- \sum_{k=0 }^{i-1} F|_{ \{ k \}\times -\N }
		- F|_{\N \times \{0\}}
		+ P_i(x)
	\right) \\
	&= (1-\phi(x)y) F|_{\N \times -\N}
		+ \phi(x) y F|_{\N \times \{0\}} \\
		&\phantom{00}
		- \sum_{i=-\ell}^{-1}\alpha_i x^{-i}y \left( \sum_{k=-\ell }^{i} F|_{ \{ k \}\times -\N } \right)
		+ \sum_{i=1}^{r - 1}\alpha_i x^{-i}y \left( \sum_{k=0 }^{i-1} F|_{ \{ k \}\times -\N } +P_i(x)\right),
\end{align*}
where $P_i(x)$ are Laurent polynomials to account for over- and under-counting.
Since each $U|_{ \{ k \}\times -\N }$ and $U|_{\N \times \{0\}}$ is automatic, each $F|_{ \{ k \}\times -\N }$ and $F|_{\N \times \{0\}}$ are algebraic by Part~(1) of Theorem~\ref{Christol}.
Hence
\[
	F|_{\N \times -\N}
	= \frac{
		G(x, y)
	}{
		1 - \phi(x) y
	}
\]
where $G(x, y)$ is algebraic.
Therefore $F|_{\N \times -\N}$ is algebraic, and $U|_{\N \times -\N}$ is $[p,p]$-automatic by Part~(2) of Theorem~\ref{Christol}.
\end{proof}

\begin{example}
Consider the Ledrappier cellular automaton with $\phi(x)=1+x^{-1}$, and let 
\begin{align*}
	L_{1} &= \N \times \{0\} \\
	 L_{2} &= \{0\} \times -\N
\end{align*}
so that $U|_{L_{1} \cup L_{2}}$ determines $U|_{\N\times -\N}$ for $\Phi$.

We have $U_{m,n}+U_{m+1,n}-U_{m,n+1}=0$ for each $(m,n)\in \Z \times \Z$, so, following the proof and notation of Theorem~\ref{2D-STDs_are_automatic}, we have
\[
	0
	= F|_{\N\times -\N} - y \left( F|_{\N\times -\N} - F|_{ L_1} \right)
	- x^{-1}y\left( F|_{\N\times -\N} - F|_{ L_1}- F|_{ L_2}+U_{0,0} \right)
\]
and therefore
\[ F|_{\N\times -\N} = \frac{ x^{-1}yU_{0,0} -(1+x^{-1})y F|_{ L_1} -x^{-1}yF|_{ L_2}}{1-(1+x^{-1})y}.
\]
If $F|_{ L_1}$ and $F|_{ L_2}$ are both algebraic, then $F|_{\N\times -\N}$ is also.
\end{example}

As we converted the $[p, p]$-kernel to the $[-p, p]$-kernel in Theorem~\ref{STD Eilenberg}, one can also convert the $[-p, p]$-kernel of a spacetime diagram in Theorem~\ref{2D-STDs_are_automatic} to the $[-p, -p]$-kernel.
For example, this enables one to compute a $[-p, -p]$-automaton for the spacetime diagram in Figure~\ref{Z x Z}.

\section{Invariant sets for linear cellular automata}\label{dynamics}

In this section and the next we apply the automaticity of spacetime diagrams, as shown in Corollary~\ref{STD is automatic} and Theorem~\ref{2D-STDs_are_automatic}, to two related questions in symbolic dynamics. We consider the $\Z\times \Z$-dynamical system $(\F_p^\Z,\sigma, \Phi)$ generated by the left shift map $\sigma$ and a linear cellular automaton $\Phi$, and we find closed subsets of $\F_p^\Z$
which are invariant under both $\sigma$ and $\Phi$. In Section~\ref{Invariant measures for linear cellular automata} we find nontrivial measures $\mu$ on $\F_p^\Z$ that are invariant under the action of $\sigma$ and $\Phi$.

By a simple transfer principle, these questions can be approached by considering dynamical systems generated by spacetime diagrams $U$ for $\Phi$.
Given a spacetime diagram $U$, one considers the subshift $(X_U,\sigma_1,\sigma_2)$, a $\Z\times \Z$-dynamical system generated by $U$; this is defined in Section~\ref{basic}. If $U$ is automatic, then $X_U$ is small in the sense of Theorem~\ref{complexity}.

The maps $\sigma$ and $\Phi$ do not exhibit the topological rigidity that Furstenberg's setting yields, as mentioned in the Introduction.
An example of a $(\sigma, \Phi)$-invariant set was first pointed out by Kitchens and Schmidt~\cite[Construction 5.2]{Kitchens--Schmidt-1992} and elaborated by Einsiedler~\cite{Einsiedler-2004}.
In Theorem~\ref{invariant} we identify a large family of $(\sigma, \Phi)$-invariant sets, and we
discuss the relationship between our invariant sets and those that are obtained by the method in \cite{Kitchens--Schmidt-1992}.

\subsection{Subshifts generated by $[-p,-p]$-automatic spacetime diagrams}\label{basic}

In this section we set up the necessary background, define subshifts generated by a spacetime diagram, and show that the subshift generated by an automatic spacetime diagram is small but infinite.
We also define substitutions, linking them to automaticity.

We equip $\F_p$ with the discrete topology and the sets $\F_p^\Z$ and $\F_p^{\Z \times \Z}$ with the metrisable product topology, noting that with this topology they are compact.
Let $\sigma_1: \F_p^{\Z\times \Z}\rightarrow \F_p^{\Z\times \Z}$ denote the left shift map $(U_{m,n})_{(m,n) \in \Z \times \Z} \mapsto (U_{m+1,n})_{(m,n) \in \Z \times \Z}$,
and let $\sigma_2: \F_p^{\Z\times \Z}\rightarrow \F_p^{\Z\times \Z}$ denote the down shift map $(U_{m,n})_{(m,n) \in \Z \times \Z} \mapsto (U_{m,n+1})_{(m,n) \in \Z \times \Z}$.
With the notation of Section~\ref{Linear cellular automata}, applying the left shift (down shift) to a sequence is equivalent to multiplying its generating function by $x^{-1}$ ($y^{-1}$).

\begin{definition}
Let $S$ and $T$ be transformations on $X$.
A set $Z\subset X$ is {\em $T$-invariant} if $T(Z)\subset Z$, and $Z$ is \emph{$(S,T)$-invariant} if it is both $S$- and $T$-invariant.
A \emph{(two-dimensional) subshift} $(X, \sigma_1,\sigma_2)$ is a dynamical system with $X$ a closed, $\sigma_1$- and $\sigma_2$-invariant subset of $\F_p^{\Z\times \Z}$.
\end{definition}

We can similarly define a one-dimensional subshift $(X, \sigma)$: here $X$ is a closed, $\sigma$-invariant subset of $\F_p^\Z$ and $\sigma$ is the left shift map.
We call $X$ the {\em shift space}.
 
Let $S\subseteq \Z \times \Z$ be a rectangle $[m_1, m_2] \times [n_1, n_2]$.
A \textit{word on $S$} is a map $w : S \to \F_p$.
These words are higher-dimensional analogues of words in one dimension, i.e.\ those indexed by a finite interval in $\Z$.
If $U\in \F_p^{\Z\times \Z}$, then $U|_S$ is the word $(U_{m, n})_{(m, n) \in S}$, and we say that the word $U|_S$ \textit{occurs} in $U$. 
Given a configuration $U\in \F_p^{\Z\times \Z}$, the \emph{language~$\mathcal{L}_U$} of $U$ is the set of all words that occur in $U$.
The \emph{language~$\mathcal{L}_X$} of a shift space $X$ is the set of all words that occur in some configuration $U\in X$.
A {\em subword} of the word $w : S \to \F_p$ is a restriction of $w$ to some rectangular $S' \subseteq S$.
The language $\mathcal L_X$ is closed under the taking of subwords, and every word in the language is extendable to a configuration in $X$.
Conversely, a language~$\mathcal{L}$ on~$\F_p$ which is closed under the taking of subwords defines a (possibly empty) subshift $(X_\mathcal{L}, \sigma_1,\sigma_2)$, where $X_\mathcal{L}$ is the set of configurations all of whose subwords belong to~$\mathcal L$.

Note that we can also define the language of an $\N\times \N$- or $\Z\times \N$-configuration $U$ and, in an analogous manner, of the $\Z\times \Z$-subshift $(X_U,\sigma_1, \sigma_2)$.

Let $U$ be a two-dimensional configuration. Recall the complexity function $c_U:\N\times \N \rightarrow \N$, where $c_U(m,n)$ is the number of distinct $m \times n$ words that occur in $U$.
We remark that the second statement of the following theorem can be improved but is sufficient for our purposes.

\begin{theorem}\label{complexity}
\leavevmode
\begin{enumerate}
\item
If the sequence $U\in \F_p^{\N \times \N}$ is $[p, p]$-automatic, then for some $K$, its complexity function satisfies $c_U(m,n) \leq K \max\{m,n \}^2$.
\item
If the sequence $U\in \F_p^{\Z \times \Z}$ is $[-p,-p]$-automatic, then for some $K$, its complexity function satisfies $c_U(m,n) \leq K \max\{m,n \}^{10}$.
\end{enumerate}
 \end{theorem}

\begin{proof}
The proof of Part~(1) is in \cite[Corollary~14.3.2]{Allouche--Shallit:2003}.
See also \cite{Allouche-Berthe-1997} and \cite{Berthe-2000}.

To see Part~(2), we recall first that, by Proposition~\ref{-p iff p}, each of $U|_{\pm \N\times \pm \N}$ is $[p,p]$-automatic, so by Part~(1), for each of them there exists a constant $K_{\pm \N \times \pm \N}$ such that $c_{U|_{\pm \N\times \pm \N}}(m,n) \leq K_{\pm \N\times \pm \N} \max\{m,n \}^2$. Let $K^*$ be the maximum of the four constants $K_{\pm \N\times \pm \N}$ and let $K\colonequal(K^*)^4$.
Let $w$ be a rectangular $m \times n$ word that occurs in $U$. If each occurrence of $w$ is entirely contained in one of the quadrants $\pm \N\times \pm \N$, then $w$ is counted by the complexity of $U$ restricted to that quadrant, and this count 
is bounded above by $ K \max\{m,n \}^2$.
Otherwise, either $S$ is partitioned into two rectangles, each of which lies in a distinct quadrant, or $S$ is partitioned into four rectangles lying in distinct quadrants. The worst case is when $S$ is a concatenation of four subrectangles, so we assume this.
There are at most $K\sum_{i=1}^{m}\sum_{j=1}^{n} \max\{ i,j\}^2 \max\{ i,n-j\}^2\max\{ m-i,j\}^2\max\{m-i,m-j\}^2$ of these subrectangles, and a crude upper estimate tells us that there are at most $K\max\{m,n\}^{10}$ such words.
\end{proof}

Theorem~\ref{complexity} tells us the languages generated by $[-p,-p]$-automatic configurations are small. On the other hand, provided that the initial conditions generating $U$ are not periodic, we now also show that they are not too small.

Let $f_u(x) = \sum_{m \in \Z} u_m x^m$ be the generating function of $u\in \F_p^\Z$ and let $F_U(x) = \sum_{m \in \Z, n\in \Z} U_{m,n} x^m y^n$ be the generating function of $U\in \F_p^{\Z\times \Z}$.
Recall that the configuration $u $ is \emph{periodic} if $ x^{-i} f_u(x)=f_u(x)$ for some $i \ge 1$ and {\em nonperiodic} otherwise.
Similarly the configuration $U $ is {\em periodic} if there exists $(i, j) \neq (0, 0)$ such that $x^{-i} y^{-j} F_U(x,y) = F_U(x,y)$ and {\em nonperiodic} otherwise.
We say that $(u_m)_{m \geq 0}$ is {\em eventually periodic} if $(x^{-i} f_u(x))|_\N$ is periodic for some $i \geq 0$.

\begin{proposition}\label{nonperiodic}
Let $u \in \F_p^\Z$ be $(-p)$-automatic,
let $\Phi:\F_p^\Z \rightarrow \F_p^\Z$ be a linear cellular automaton whose generating polynomial is neither $0$ nor a monomial, and let $U\in \F_p^{\Z \times \Z}$ be a spacetime diagram for $\Phi$ with $U|_{\Z \times \{ 0\}}=u$.
If $(u_m)_{m\geq 0} $ is not eventually periodic, then $U$ is nonperiodic.
\end{proposition}

\begin{proof}
Suppose that $U$ is periodic.
Then there is $(i, j) \neq (0, 0)$ such that $x^{-i} y^{-j} F_U(x, y) = F_U(x, y)$.
We can assume without loss of generality that $-j \geq 0$.
We have $x^i F_U(x,y) = y^{-j} F_U(x,y)$. Restricting to $\Z\times \{ 0\}$, we get $x^i f_{u}(x) = \phi(x)^{-j} f_{u}(x)$, where $\phi(x)$ is the generating polynomial of $\Phi$.
In other words $\left(\phi(x)^{-j} - x^i\right) f_u(x) = 0$, where by assumption $\phi(x)^{-j} - x^i \neq 0$.
Thus $(u_m)_{m \geq \min\{i, r j\}}$ satisfies a linear recurrence and hence is eventually periodic.
\end{proof}

\begin{corollary}
Under the conditions of Proposition~\ref{nonperiodic},
if $(u_m)_{m\geq 0} $ is not eventually periodic, then $c_U(m, n) > m n$ for each $m$ and $n\in \N$.
\end{corollary}

\begin{proof}
This follows directly from \cite[Corollary~9 and the remark following it]{Kari-Moutot}, where Kari and Moutot show that Nivat's conjecture holds for $\Z \times \Z$-indexed spacetime diagrams $U$ of a linear cellular automaton:
If $c_U(m,n) \leq mn$ for some $m$ and $n$, then $U$ is periodic.
\end{proof}

We remark that in \cite{Quas-Zamboni-2004} and \cite{Cyr-Kra-2015} there are more general but less sharp results concerning Nivat's conjecture.

Let $\Phi:\F_p^\Z\rightarrow \F_p^\Z$ be a linear cellular automaton, and let $U$ in $\F_p^{\Z \times \Z}$ or $\F_p^{\Z \times \N}$ be a spacetime diagram for $\Phi$. Define 
\[
	X_{U}\colonequal\{ V\in \F_p^{\Z\times \Z}: \mathcal L_V\subseteq \mathcal L_U \}.
\]
We call $(X_{U}, \sigma_1, \sigma_2)$ the {\em $\Z \times \Z$-subshift defined by $U$}. 
We consider spacetime diagrams $U\in \F_p^{\Z\times \Z}$ which are $[-p,-p]$-automatic. By Theorem~\ref{2D-STDs_are_automatic}, we obtain these once we choose automatic sequences as initial conditions, in $U|_{\{i\}\times -\N}$, for $-\ell \leq i \leq r-1$, in $U|_{-\N \times \{0\}}$, and in $U|_{\N \times \{0\}}$.

\begin{lemma}\label{shift elements are spacetime diagrams}
Let $\Phi:\F_p^\Z\rightarrow \F_p^\Z$ be a linear cellular automaton, let $U\in \F_p^{\Z \times \Z}$ be a spacetime diagram for $\Phi$, and let $(X_{U},\sigma_1,\sigma_2)$ be the $\Z \times \Z$-subshift defined by $U$. Then every element of $ X_{U}$ is a spacetime diagram for $\Phi$.
\end{lemma}

\begin{proof}
Let $\phi(x) = \alpha_{-\ell}x^{\ell} + \dots +\alpha_0 + \dots +\alpha_r x^{-r}$ be the generating polynomial of $\Phi$.
If some element $V \in X_U$ is not a spacetime diagram for $\Phi$, then $\Phi$'s local rule is violated somewhere, i.e.\ for some $m,n$ we have $ \alpha_{-\ell}V_{m,n} + \dots + \alpha_0 V_{m+\ell,n}+\dots +\alpha_r V_{m+\ell+r-1,n}\neq V_{m+\ell, n+1}$. By definition the rectangular word $w\colonequal(V_{i,j})_{m\leq i\leq m+\ell+r-1, n\leq j \leq n+1}$ belongs to the language of $U$; that is, $w$ occurs in $U$ and agrees with $\Phi$'s local rule, a contradiction.
\end{proof}

We collect some facts about constant-length substitutive sequences, 
referring the reader to \cite{Allouche--Shallit:2003} for a thorough exposition. 
A {\em substitution of length $p$} is a map $\theta:\mathcal A \rightarrow \mathcal A^p$.
We use concatenation to 
extend $\theta$ to a map on finite and infinite words 
from $\mathcal A$. 
By iterating $\theta$ on any fixed letter $a \in \mathcal A$, 
we obtain infinite configurations $u\in\mathcal A^\N$
such that $\theta^j(u)=u$ 
for some natural number $j$; we call such configurations {\em $\theta$-periodic}, or {\em $\theta$-fixed} if $j=1$.
We write $\theta^\infty(a)$ to denote a fixed point.
The pigeonhole principle implies that $\theta$ has a $\theta$-periodic configuration.
We can also define bi-infinite fixed points of $\theta$.
Given a bi-infinite sequence $u= \cdots u_{-2}u_{-1}\cdot u_{0} u_{1} \cdots \in \mathcal A^{\Z}$ and substitution $\theta$ on $\mathcal A$, define $\theta(u) = \cdots \theta(u_{-2})\theta (u_{-1})\cdot \theta (u_0) \theta (u_1)\cdots$.
If $a, b$ are letters such that $\theta(a)$ starts with $a$, $\theta(b)$ ends with $b$, and the word $ba$ occurs in $\theta^n(c)$ for some letter $c$, then we call the unique sequence $u= \cdots b \cdot a \cdots$ that satisfies $\theta(u)=u$ a {\em bi-infinite fixed point of $\theta$}.
Bi-infinite fixed points of a length-$p$ substitution $\theta$ are $(-p)$-automatic, since $p$-automatic sequences are closed under shifting to the right and the addition of finitely many new entries; see \cite[Theorem~6.8.4]{Allouche--Shallit:2003}.

We can similarly define two-dimensional substitutions $\theta:\mathcal A \rightarrow \mathcal A^{p\times p}$ and two-dimensional $\theta$-fixed points.

We recall Cobham's theorem~\cite{Cobham-1972}.
We refer to \cite[Theorems~6.3.2 and 14.2.3]{Allouche--Shallit:2003} for the proof.

\begin{theorem}\label{Cobham}
\leavevmode
\begin{enumerate}
\item
The sequence $(u_{m})_{m\geq 0}\in \F_p^\N$ is $p$-automatic if and only if it is the image, under a coding, of a fixed point of a length-$p$ substitution $\theta$.
\item
The sequence $(U_{m,n})_{m\geq 0, n\geq 0}\in \F_p^{\N\times \N}$ is $[p,p]$-automatic if and only if it is the image, under a coding, of a fixed point of a substitution $\theta:\mathcal A \rightarrow \mathcal A^{p\times p}$.
\end{enumerate}
\end{theorem}

\begin{example}\label{p=3 example part 3}
As in Examples~\ref{p=3 example} and \ref{p=3 example part 2}, let $p = 3$, and let $\phi(x) = x + 1 \in \F_3[x]$.
We perform a search to find substitutions $\theta : \F_3 \to \F_3^3$ with fixed points $\theta^\infty(a)$ generated by small automata under Part~(1) of Theorem~\ref{Cobham}, since a small automaton makes subsequent computations feasible.
We also require that $\theta$ is primitive, that the fixed point $(u_m)_{m \geq 0}$ is not eventually periodic, and that $(u_{3 m})_{m \geq 0}$, $(u_{3 m + 1})_{m \geq 0}$, and $(u_{3 m + 2})_{m \geq 0}$ are not eventually periodic.
Among the substitutions satisfying these criteria, the substitution $\theta$ defined by $\theta(0) = 001$, $\theta(1) = 112$, and $\theta(2) = 220$ minimizes the number of states in the corresponding automaton, producing the automaton on the left in Figure~\ref{automata} for the fixed point $\theta^\infty(0)$.
Indeed this is how we chose that automaton.
From the $54$-state automaton for $U|_{\N \times \N}$, we compute by Part~(2) of Theorem~\ref{Cobham} a substitution $\Theta : \mathcal A \to \mathcal A^{3 \times 3}$ and coding $\tau : \mathcal A \to \F_3$ such that $\tau(\Theta^\infty(a)) = U|_{\N \times \N}$ for a particular letter $a \in \mathcal A$.
The size of the alphabet is $|\mathcal A| = 75$.
\end{example}

Note that while the spacetime diagram has a substitutional nature, the alphabet size makes the computation of this substitution by hand infeasible.
This is presumably why such substitutions have not been studied in the symbolic dynamics literature.

\subsection{Automatic invariant sets and intersection sets}\label{invariant sets}

For a linear cellular automaton $\Phi:\F_p^\Z\rightarrow \F_p^\Z$, let
\[
	X_\Phi = \{V \in \F_p^{\Z \times \Z}: \text{$V$ is a spacetime diagram for $\Phi$}\}.
\]
Then $X_\Phi$ is closed in $\F_p^{\Z\times \Z}$ and $(X_\Phi, \sigma_1, \sigma_2)$ is a $\Z \times \Z$-subshift, an example of a {\em Markov subgroup} or {\em algebraic shift}~\cite{Schmidt-1995}.

We define $\pi : X_\Phi \rightarrow \F_p^{\Z}$ by $\pi(V) = V|_{\Z \times \{0\}}$.
Let $Z \subset X_\Phi$ be a closed and $(\sigma_1,\sigma_2)$-invariant subset.
Note that by construction $\Phi$ maps $ \pi(Z)$ onto $ \pi(Z)$, though $\Phi $ is not necessarily invertible on $ \pi(Z)$; i.e.\ we have two commuting transformations $\sigma$ and $ \Phi$ defined on $ \pi(Z)$ that define a monoid action of $\Z \times \N$.
The reader who prefers to work with a $\Z \times \Z$ action can take the natural extension of $( \pi(Z), \sigma, \Phi)$; see for example the exposition in \cite{Cyr--Kra-2017}.
We have
\begin{equation}\label{commute}
	\text{$\pi\circ \sigma_1=\sigma \circ \pi$ \quad and \quad $\pi\circ \sigma_2=\Phi\circ \pi$}.
\end{equation}
 
\begin{theorem}\label{invariant}
Let $\Phi:\F_p^\Z\rightarrow \F_p^\Z$ be a linear cellular automaton whose generating polynomial is neither $0$ nor a monomial, and let $u\in \F_p^{\Z}$ be a $(-p)$-automatic sequence which is not eventually periodic. Then $\pi(X_{\ST_\Phi (u)})$ is a closed $(\sigma,\Phi)$-invariant subset of $\F_p^\Z$ which is neither finite nor equal to $\F_p^\Z$.
\end{theorem}

\begin{proof}
By the identities in (\ref{commute}), any closed $(\sigma_1,\sigma_2)$-invariant set in $X_\Phi$ projects to a closed $(\sigma,\Phi)$-invariant subset of $\F_p^\Z$.
Thus $\pi(X_{\ST_\Phi (u)})$ is $(\sigma,\Phi)$-invariant, and compactness implies that it is closed in $\F_p^\Z$. By Proposition~\ref{nonperiodic}, $\pi(X_{\ST_\Phi (u)})$ is not finite. By Theorem~\ref{complexity}, $\pi(X_{\ST_\Phi (u)})\neq \F_p^\Z$.
\end{proof}

There are other examples of invariant sets for linear cellular automata.
This was first touched on by Kitchens and Schmidt~\cite[Construction 5.2]{Kitchens--Schmidt-1992} \cite[Example 29.8]{Schmidt-1995} and by Silberger~\cite[Example 3.4]{Silberger-1995}, where the following construction is described.
One starts with a finite set $H\subset \F_p^{j}$ and considers $H^\Z$.
There is a natural injection $i:H^\Z\rightarrow \F_p^\Z$ obtained by concatenating. Note that $i(H^\Z)$ is not necessarily invariant under the left shift $\sigma$, but $\bar Y \colonequal \cup_{m=0}^{j-1}\sigma^m (i(H^\Z))$ is. 
It is clear that $\bar Y$ is a proper subset of $\F_p^\Z$.
However, to extend $\bar Y$ to a ``small" set which is invariant under $\Phi$, Kitchens and Schmidt~\cite[Construction 5.2]{Kitchens--Schmidt-1992} assume in addition that $H$ is a group and that $j$ has a simple base-$p$ representation. For example, they take $j=p^k$, and then the assumption that $H=H_k$ is a group and the ``freshman's dream'' (which is that if $\Phi$ has generating polynomial $\phi(x) =\alpha_{-\ell}x^{\ell} + \dots +\alpha_0 + \dots +\alpha_r x^{-r}$ then $\Phi^{p^k}$ has generating polynomial $\phi(x)^{p^k} =\alpha_{-\ell}x^{\ell p^k} + \dots +\alpha_0 + \dots +\alpha_r x^{-r p^k}$) imply that $\Phi^{p^k}(\bar Y_{k})\subseteq \bar Y_{k}$.
Therefore $Y_{k}\colonequal\cup_{n=0}^{p^k - 1}\Phi^{n}(\bar Y_{k})$ is $(\sigma,\Phi)$-invariant and is also a proper subset of $\F_p^\Z$. 
One can also obtain more complex subshifts by taking an infinite intersection $\cap_k Y_k$ of nested shift spaces where $Y_k$ is built from a group $H_k\subset\F_p^{p^k}$ and $k\rightarrow \infty$.

\begin{example}\label{intersection set}
Let $p=2$, let $\Phi$ be the Ledrappier cellular automaton, and let $H_k=\{0^{2^k}, \theta^k(0), \theta^k(1), 1^{2^k} \}$ where $\theta $ is the Thue--Morse substitution. Then, using the freshman's dream, $\cap_k Y_k$ contains $\pi(X_{{\ST_\Phi (u)}})$, where $u\in \F_p^\Z$ is any bi-infinite fixed point of the Thue--Morse substitution.
Note that in fact here $\pi(X_{\ST_\Phi (u)})$ is almost all of $\cap_k Y_k$, as $\cap_k Y_k \backslash \pi(X_{\ST_\Phi (u)})$ consists of bi-infinite sequences which are identically $0$ to the left of some index and which are a $\theta$-fixed point to the right of that index, or vice versa. We can rectify this discrepancy by changing our initial condition. 
If one starts with the $(-2)$-automatic initial condition $u$ whose right half is a fixed point of $\theta$ and whose left half is identically $0$, then $\pi(X_{\ST_\Phi (u)})= \cap_k Y_k$.
\end{example}

This construction is explored in greater detail by Einsiedler~\cite{Einsiedler-2004}, who shows that one can find $(\sigma_1,\sigma_2)$-invariant sets of any possible entropy. His construction is based on the construction of Kitchens and Schmidt, although he expresses it differently. Precisely, recall that $X_\Phi$ is the set of all spacetime diagrams for $\Phi$.
Einsiedler works with a group $Z\subset X_\Phi$ which is invariant under the action of some $\sigma_1^m \sigma_2^n$. For example, 
if one considers the group
\[
	Z \colonequal \{V \in X_\Phi: \text{$V_{2m,2n}=0$ for each $m,n\in \Z$}\},
\]
then this group is invariant under $\sigma_1^2 \sigma_2^2$. Using the Kitchens--Schmidt construction, it can be generated by taking spacetime diagrams of sequences on $H=\{(0,0), (1,1) \}\in \F_2^2$
with the Ledrappier cellular automaton $\Phi$. For, the image of a sequence in $H^\Z$ under $\Phi$ contains a $0$ in every even index, and the image of a sequence in $H^\Z$ under $\Phi^2$ is a sequence in $H^\Z$. Einsiedler also allows addition of $Z$ by a finite set $F$.
He calls sets $Z=\cap_k (Z_k +F_k)$ {\em intersection sets}, and he asks whether there is a description of every $(\sigma_1,\sigma_2)$-invariant set in terms of intersection sets.

\begin{theorem}\label{intersection theorem}
Let $\Phi:\F_p^{\Z}\rightarrow \F_p^\Z$ be a linear cellular automaton, and let $u\in \F_p^{\Z}$ be a $(-p)$-automatic sequence which is not eventually periodic. Then $\pi(X_{\ST_\Phi (u)})$ is a $(\sigma,\Phi)$-invariant proper subset of $\F_p^\Z$ which is a subset of an intersection set.
\end{theorem}

\begin{proof}
By assumption, $u$ is a concatenation of two $p$-automatic sequences. By Cobham's theorem, there are substitutions $\theta_1:\mathcal A_1 \rightarrow \mathcal A_1^p$ and $\theta_2: \mathcal A_2\rightarrow \mathcal A_2^p$, and codings $\tau_1: \mathcal A_1\rightarrow \F_p$ and $\tau_2:\mathcal A_2\rightarrow \F_p$ such that $u|_{\N}$ is the $\tau_1$-coding of a right-infinite fixed point of $\theta_1$, and $u|_{-\N }$ is the $\tau_2$-coding of a left-infinite fixed point of $\theta_2$.
For each $k$ let $H_k$ be the group in $\F_p^{p^k}$ generated by $\{\tau_1(\theta_1^k(a)): a \in \mathcal A_1\} \cup \{\tau_2(\theta_2^k(a)): a \in \mathcal A_2\}$. Let $Y_k$ be the $(\sigma,\Phi)$-invariant subset of $\F_p^\Z$ as defined above using the group $H_k$. Then for each $k$, $\pi(X_{\ST_\Phi (u)}) \subset Y_k$, so $\pi(X_{\ST_\Phi (u)})\subset \cap_k Y_k$.
\end{proof}

In Example~\ref{intersection set}, we can find $u$ such that the set $\pi(X_{\ST_\Phi (u)}) $ is equal to an intersection set $\cap_k Y_k$. This is because for each $k$ the group generated by $\{\theta^k(0), \theta^k(1)\}$ is very close to the set $\{\theta^k(0), \theta^k(1)\}$.

\begin{example}\label{p=3 example part 4}
We continue with our running example, last seen in Example~\ref{p=3 example part 3}, where $p=3$, $\Phi$ is the cellular automaton with generating function $x+1$, and the initial condition is generated by the substitution $\theta(0)=001$, $\theta(1)=112$, $\theta(2)=220$. Every word of length $2$ occurs in every fixed point of $\theta$.
One shows by induction that
\begin{equation}\label{total sum} \theta^k(0) + \theta^k(1) + \theta^k(2) = 0^{3^k} \end{equation}
for each $k$.
We also have
\begin{equation}\label{group sum}
	2\theta^k(0)+\theta^k(1) = 2\theta^k(1)+\theta^k(2)= 2\theta^k(2)+\theta^k(0)=1^{3^k},
\end{equation}
so that the group generated by $\{ \theta^k(0), \theta^k(1), \theta^k(2) \}$ is
\[ H_k = \{ 0^{3^k}, 1^{3^k}, 2^{3^k}, \theta^k(0), 2\theta^k(0), \theta^k(1), 2\theta^k(1), \theta^k(2), 2\theta^k(2)\}.\]
Let $(u_m)_{m \geq 0}$ be the fixed point $\theta^\infty(0)$ and let $(u_{-m})_{m \geq 0}$ be the constant $0$ sequence. Its spacetime diagram $\ST_\Phi(u)$ is shown in Figure~\ref{p=3 right part}.
We claim that all words in $H_k$ occur horizontally in $\ST_\Phi (u)$.
The words $0^{3^k}$, $\theta^k(0)$, $\theta^k(1)$, and $\theta^k(2)$ occur in the $0$-th row of $\ST_\Phi(u)$.
Since all possible words of length $2$ occur in $u$, each element of
\[
	S_k= \{ \theta^k(a)+\theta^k(b): ab \in \F_3\times \F_3 \}
	= \{2\theta^k(0), 2\theta^k(1), 2\theta^k(2) \}
\]
occurs in the $3^k$-th row of $\ST_\Phi(u)$.
Also, since $(x + 1)^{4 \cdot 3^k} = x^{4 \cdot 3^k} + x^{3 \cdot 3^k} + x^{3^k} + 1$, Equation~\eqref{group sum} implies
\begin{align*}
	\Phi^{4 \cdot 3^k}(u)|_{[3 \cdot 3^k, 4 \cdot 3^k - 1]}
	&= u|_{[3 \cdot 3^k, 4 \cdot 3^k - 1]}
	+ u|_{[2 \cdot 3^k, 3 \cdot 3^k - 1]}
	+ u|_{[0, 3^k - 1]}
	+ u|_{[-3^k, -1]} \\
	&= \theta^k(0) + \theta^k(1) + \theta^k(0) + 0^{3^k} = 1^{3^k}.
\end{align*}
It follows that $2^{3^k -1}$ occurs in row $4 \cdot 3^k +1$; this is true for all $k$, so $2^{3^k}$ also occurs.
Therefore all words in $H_k$ occur in $\ST_\Phi (u)$, and by approximation arguments one sees that $\pi(X_{\ST_\Phi(u)})=\cap_k Y_k$.
\end{example}

In contrast, for the initial condition $u$ in Figure~\ref{p=3 both parts}, it is not so clear that $\pi(X_{\ST_\Phi(u)})$ is an intersection set. In Example~\ref{p=3 example part 5}, for a different initial condition $u$, which is also not eventually periodic in either direction, we describe $\pi(X_{\ST_\Phi(u)})$ as a modified intersection set $\cap_k Y_k$, where $Y_k$ is defined with sets of words $H_k$ which are not groups, but which nevertheless capture the words we see at levels $p^k$.

\begin{question}
Can all of the invariant sets in Theorem~\ref{invariant} be written as intersection sets?
\end{question}

\section{Invariant measures for linear cellular automata}\label{Invariant measures for linear cellular automata}

In this section we study the $(\sigma,\Phi)$-invariant measures that are supported on the invariant sets found in Theorem~\ref{invariant}.
By the same transfer principle mentioned in Section~\ref{dynamics}, a measure supported on $X_U$ that is invariant under $\sigma_1$ and $\sigma_2$ transfers to a measure on $\F_p^\Z$ which is invariant under $\sigma $ and $\Phi$. By Proposition~\ref{not Lebesgue}, these measures are never the Haar measure.
In Theorem~\ref{decidability} we identify a decidable condition which guarantees that the measure $\mu$ in question is not finitely supported, and in 
Theorem~\ref{power-free} we identify a family of nontrivial $(\sigma,\Phi)$-invariant measures when $\Phi$ is the Ledrappier cellular automaton.
In Theorem~\ref{nature of measures} we identify $(\sigma,\Phi)$-invariant measures as belonging to simplices whose extreme points are ergodic measures supported on codings of substitutional shifts. This statement implicitly contains another method by which to determine whether $\mu$ is trivial, as there exist algorithms to compute the frequency of a word for such a measure.
Finally, in Theorems~\ref{coincidence} and \ref{unique letter at nonzero coefficients}, we give conditions that guarantee that the shifts we study contain constant configurations and hence possibly lead to finitely supported $(\sigma,\Phi)$-invariant measures.

Throughout this section, we make use of the substitutional characterisation of automatic sequences to state and prove our results.

\subsection{Invariant measures on $[-p,-p]$-automatic spacetime diagrams}\label{measure}

Recall that a subshift $(X,\sigma)$ is \emph{aperiodic} if each $x \in X$ is aperiodic. We consider measures on the Borel $\sigma$-algebra of $X$.
Let $S,T:X\rightarrow X$ be transformations on $X$. 
A measure $\mu$ on $X$ is {\em $T$-invariant} if $\mu(Z)= \mu(T^{-1}(Z)) $ for every measurable $Z$, and 
 it is {\em $(S,T)$-invariant} if it is both $S$- and $T$-invariant.
A measure $\mu$ has {\em finite support $\{ x_1, \dots, x_n\}$} if it is a finite weighted sum of Dirac measures $\mu = \sum_{i=1}^nw_i \delta_{x_i}$. If the finitely-supported Borel measure $\mu$ on a shift space $X\subseteq \F_p^\Z$ is also $\sigma$-invariant, then each configuration in the support of $\mu$ is periodic. The same is true if $\mu$ is finitely supported on a two-dimensional shift space and is $(\sigma_1,\sigma_2)$-invariant. In the next proposition we list some elementary observations about the measures on $Y_U$ that are projections of measures on $X_U$. By the Krylov--Bogolyubov theorem~\cite[Theorem 6.9]{Walters-1982}, there exist $(\sigma_1,\sigma_2)$-invariant measures supported on $X_U$. Recall that the map $\pi : X_U \rightarrow \F_p^{\Z}$ is defined by $\pi(V) = V|_{\Z \times \{0\}}$.

\begin{proposition}\label{not Lebesgue}
Let $\Phi:\F_p^\Z\rightarrow \F_p^\Z$ be a linear cellular automaton, and let $U \in \F_p^{\Z\times \Z}$ be a $[-p,-p]$-automatic spacetime diagram for $\Phi$.
Let $(Y_U,\sigma) $ be the $\Z$-subshift defined by $U$.
Let $\mu$ be a $(\sigma_1,\sigma_2)$-invariant measure on $X_U$, and let $\lambda\colonequal\mu\circ \pi^{-1}$.
\begin{enumerate}
\item
Then $\lambda$ is a $(\sigma,\Phi)$-invariant measure on $Y_U$ that is not the Haar measure. 
\item
Moreover, if $\mu$ is not finitely supported, then $\lambda$ is not finitely supported.
\end{enumerate}
\end{proposition}

\begin{proof}
By Equations~\eqref{commute}, any Borel measure $\mu$ on $ X_U$ which is $(\sigma_1,\sigma_2)$-invariant defines a $(\sigma, \Phi)$-invariant Borel measure $\lambda\colonequal\mu\circ\pi^{-1}$ on $Y_U$.
By Part~(2) of Theorem~\ref{complexity}, there is a $K$ such that there are at most $Km^{10}$ words on an $m\times 1$ rectangle in $\mathcal L_{ U}$, so there are at most $Km^{10}$ words of length $m$ in the language of $Y_U$. Thus for large $m$, there exists a word $w$ of length $m$ such that $\lambda(w)=0$. This proves the first assertion.

To see the second assertion, if $\lambda$ is supported on a finite set $\{y_1, \dots, y_n\}$, then, as $\lambda$ is invariant under $\Phi^{-1}$, for each $i$ we have $\Phi^{-1}(y_{i}) \cap \{y_1, \dots, y_n \}\neq \emptyset$.
For each $i$, this implies that $\Phi^{-1}(y_{i}) \cap \{y_1, \dots, y_n \}$ consists of exactly one element.
Therefore $\Phi$ is a permutation on $\{y_1, \dots, y_n\}$.
For each cycle in this permutation, consider the $\Z \times \Z$-configurations whose rows are elements of the cycle.
Then $\mu$ is supported on the union of the $(\sigma_1, \sigma_2)$-orbits of these $\Z \times \Z$-configurations.
Since $\lambda$ is invariant under the left shift, each $y_{i}$ is periodic.
Therefore $\mu$ is finitely supported.
\end{proof}

In the following theorem we give a condition that guarantees the existence of measures on $Y_U$ which are $(\sigma,\Phi)$-invariant and which are not finitely supported.
We say that a two-dimensional configuration $U$ is {\em horizontally $M$-power-free} if no $m \times 1$ word of the form $w^M$ with $m \geq 1$ occurs in $U$.

\begin{theorem}\label{decidability}
Let $U \in \F_p^{\Z\times \Z}$ be a $[-p,-p]$-automatic sequence, specified by an automaton.
It is decidable whether there exists $M \geq 2$ such that $U$ is horizontally $M$-power-free.
\end{theorem}

\begin{proof}
We reduce the decidability of horizontal $M$-power-freeness of $U$ to that of each quadrant.

An occurrence of a horizontal $M$-power $w^M$ with $|w| = \ell$ in the sequence $(U_{m, n})_{(m, n) \in \Z \times \Z}$ is a word of the form $U_{m, n} \cdots U_{m + M \ell - 1, n}$ satisfying $U_{i, n} = U_{i + \ell, n}$ for all $i$ in the interval $m \leq i \leq m + (M - 1) \ell - 1$.
Therefore $U$ is horizontally $M$-power-free if and only if the set
\[
	S\colonequal\{ ( M,\ell): (\exists m \geq 0) (\exists n \geq 0) (\forall i) ((0\leq i \leq M \ell - 1) \to (U_{m+i, n}= U_{m+i+\ell, n}))\}
\]
is empty.
We follow Charlier, Rampersad, and Shallit~\cite[Theorem 4]{Charlier--Rampersad--Shallit-2012}.
The configuration $U$ is horizontally $M$-power-free for arbitrarily large $M$ if and only if for all $k\geq 0$, $S$ contains a pair $(M,\ell)$ with $M >\ell p^{k}$.
Padding the shorter word with zeros if necessary, we write the base-$p$ representation of the pair $(M,\ell)$ as $(M_e, \ell_e), (M_{e -1}, \ell_{e-1}), \dots, (M_0,\ell_0)$. Thus for every $k\geq 0$, $S$ contains a pair $(M,\ell)$ with $M\geq \ell p^{k}$ if and only if $S$ contains a pair $(M,\ell)$ whose base-$p$ representation starts with $(d_1,0),(d_2,0), \dots, (d_k,0)$, where $d_1\neq 0$ and each other $d_i\in \F_p$. 
Given the automaton $\mathcal M$ which generates $\chi_S$, $S$ contains a pair $(M,\ell)$ with $M\geq \ell p^{k}$ for arbitrarily large $k$ if and only if there are words $u$, $w$, and $v$ on the alphabet $\F_p\times \F_p$ with the second entries of all letters in $w$ and $v$ all equal to $0$, and where $u$ is the label of a path from the initial state of $\mathcal M$ to a state $s$, $w$ is the label of a cycle at $s$, and $v$ is the label of a path from $s$ to a state whose corresponding output is $1$. Whether three such words exist is decidable.
\end{proof}

For fixed $M$, the set $S$ in the proof is a {\em $p$-definable set} (see \cite[Definition 6.34]{Rigo-2014}), and horizontal $M$-power-freeness can be determined by constructing an automaton; see \cite[Section~6.4]{Rigo-2014} and \cite{Mousavi}.

\begin{corollary}\label{not point mass}
Let $\Phi:\F_p^\Z\rightarrow \F_p^\Z$ be a linear cellular automaton, let $U \in \F_p^{\Z\times \Z}$ be a $[-p,-p]$-automatic spacetime diagram for $\Phi$, and let $(Y_U,\sigma)$ be the $\Z$-subshift defined by $U$.
If $U$ is horizontally $M$-power-free for some $M \geq 2$, then there exists a $(\sigma,\Phi)$-invariant measure $\lambda$ on $Y_U$ which is neither the Haar measure, nor finitely supported.
\end{corollary}

\begin{proof}
Recall that a finitely-supported $\sigma$-invariant measure $\lambda$ is supported on a set $\{y_1, y_2, \dots, y_n\} \subseteq Y_U$ where each $y_i$ is periodic.
If $X_U$ is horizontally $M$-power-free, then $Y_U$ is aperiodic.
Thus for any $(\sigma_1, \sigma_2)$-invariant measure $\mu$ on $(X, \sigma_1, \sigma_2)$, $\mu \circ \pi^{-1}$ is a $(\sigma, \Phi)$-invariant measure which is not finitely supported. By Proposition~\ref{not Lebesgue}, $\mu \circ \pi^{-1}$ is not the Haar measure.
\end{proof}

Note that if we take the initial condition $u$ to be an aperiodic fixed point of a primitive substitution, then, by results of Moss\'{e}~\cite{Mosse:92}, $u$ is $M$-power-free for some $M$.

Continuing with Example~\ref{intersection set}, Schmidt~\cite[Example 29.8]{Schmidt-1995} identifies a $(\sigma, \Phi)$-invariant measure which is supported on $\pi(X_U)$, where $\Phi$ is the Ledrappier cellular automaton, $U=\ST_\Phi(u)$, $u_m = 0$ for all $m \leq -1$, and $(u_m)_{m \geq 0}$ is a fixed point of the Thue--Morse substitution.
He does not study whether this measure is finitely supported;
our experiments suggest that this measure is a point mass supported on the constant zero configuration.
However in the next theorem we identify a family of substitutions which do yield nontrivial $(\sigma,\Phi)$-invariant measures for the Ledrappier cellular automaton.

Given a substitution $\theta:\F_p\rightarrow \F_p^p$, we write $\theta(a) = \theta_0(a) \cdots \theta_{p-1}(a)$.
We say that $\theta$ is {\em bijective} if, for each $i$ in the interval $0\leq i\leq p-1$, $\{\theta_i(a): a\in \F_p \}= \F_p$.

\begin{theorem}\label{power-free}
Let $\Phi:\F_3^\Z\rightarrow \F_3^\Z$ be the linear cellular automaton with generating polynomial $\phi(x)= x + 1$, let $\theta$ be a primitive bijective substitution on $\F_3$, and suppose that $u\in \F_3^\Z$ is a bi-infinite aperiodic fixed point of $\theta$.
Then there exists $M$ such that $\ST_\Phi(u)$ is horizontally $M$-power-free.
\end{theorem}

\begin{proof}
Since $\theta$ is bijective, $\theta$ satisfies Identity~\eqref{total sum}:
\[ 
	\theta^k(0) + \theta^k(1) + \theta^k(2) = 0^{3^k}.
\]
We claim that, for each $k\geq 1$, for each $n\geq 0$, and for each $m\in \Z$, we have
\begin{equation}\label{first claim}
	\Phi^{n\cdot 3^k}(u)|_{[m3^k, (m + 1) 3^k - 1]} \in
	\begin{cases}
		\{\theta^k(0), \theta^k(1), \theta^k(2)\}		& \text{if $n$ is even} \\
		\{2 \theta^k(0), 2 \theta^k(1), 2 \theta^k(2)\}	& \text{if $n$ is odd.}
	\end{cases}
\end{equation}
Fix $k\geq 1$.
Since $u$ is a bi-infinite fixed point of $\theta$, we have $u|_{[m3^k, (m + 1) 3^k - 1]} \in \{\theta^k(0), \theta^k(1), \theta^k(2)\}$.
Let $n = 1$.
Since $(x + 1)^{3^k} = x^{3^k} + 1$, we have
\begin{align*}
	\Phi^{3^k}(u)|_{[m3^k, (m + 1) 3^k-1]}
	&= u|_{[m3^k, (m + 1) 3^k-1]}+ u|_{[ (m-1)3^k, m 3^k - 1]} \\
	&= \theta^k(u_m) + \theta^k(u_{m-1}) \\
	&\in \{2\theta^k(0), 2\theta^k(1), 2\theta^k(2)\}
\end{align*}
for each $m\in \Z$.
The claim follows by induction on $n$ by replacing $u$ with $\Phi^{3^k}(u)$.

For each $k$, let
\[
	H_k = \{ \theta^k(0), 2\theta^k(0), \theta^k(1), 2\theta^k(1), \theta^k(2), 2\theta^k(2)\}.
\]
(Note that $H_k$ is not a group, contrary to the definition of an intersection set.)
Since $u$ is an aperiodic fixed point of a primitive substitution, Moss\'{e}'s theorem~\cite{Mosse:92} tells us that $u$ is $M$-power-free for some $M \geq 2$.
This implies that $\theta^k(a)$ is $M$-power-free for each $a\in \F_3$, and hence $2\theta^k(a)$ is also $M$-power-free.
Thus all words in $H_k$ are $M$-power-free, so 
if a power $w^l$ occurs as a subword of a word in $H_k$, then $l < M$.

Next note that, again because words in $H_k$ are $M$-power-free, if a word in $H_k$ is tiled by a word $w$ (that is, is a subword of $w^\infty$), then $|w| > \frac{3^k}{M}$.
This implies that if $|w| \leq \frac{3^k}{M}$ and $w^l$ occurs as a subword of $W_1\cdots W_{j} \in H_k^{j}$, then $w^l$ occurs as a subword of $W_i W_{i + 1}$ for some $1\leq i \leq j-1$,
 and so $l \leq 2M - 2$.

Given a word $w=w_1 \cdots w_m$ of length $m \geq 2$, define $\Phi(w) \colonequal (w_1+w_2) \cdots (w_{m -1}+w_m)$.
Suppose $w^l$ occurs in the $n$-th row of $\ST_\Phi(u)$.
We show that $l < 9 M$.
Let $k$ be such that $3^{k+1}\leq |w^l |=l |w|<3^{k+2}$.
Then $|w| < \frac{3^{k + 2}}{l}$.
Let $N$ be such that ${N\cdot 3^k}\leq n< (N+1)\cdot 3^k$.
Write $\Phi^{ (N+1)\cdot 3^k - n}(w^l)= \bar w^{\bar l} \bar v$, where the words $\bar w$ and $\bar v$ are such that $\bar l \geq 1$ is maximal and $\bar v$ is a prefix of $\bar w$ with $0\leq |\bar v|\leq |\bar w|-1$.
We have $|\bar w| \leq |w|$ since the period length of a word does not increase after applying $\Phi$.
There are two cases.

If $|\bar w| \geq \frac{3^k}{M}$, then $\frac{3^k}{M} \leq |\bar w| \leq |w| < \frac{3^{k + 2}}{l}$, so $l < 9 M$.

If $|\bar w| < \frac{3^k}{M}$, then, since $\bar w^{\bar l}$ occurs on row $(N+1) \cdot 3^k$, by \eqref{first claim} $\bar w^{\bar l}$ occurs as a subword of $W_1\cdots W_{j} \in H_k^{j}$ for some $j$.
By the argument above, $\bar w^{\bar l}$ occurs as a subword of $W_iW_{i+1}$ and therefore $\bar l \leq 2M-2$.
We also have
\begin{align*}
	|\bar w^{\bar l} \bar v| =
	|\Phi^{ (N+1)\cdot 3^k - n}(w^l)|
	&= |w^l| - \left((N+1)\cdot 3^k - n\right) \\
	&\geq 3^{k+1}- (N+1)\cdot 3^k + N \cdot 3^k \\
	&= 2\cdot 3^k,
\end{align*}
so
\[
	2\cdot 3^k \leq |\bar w^{\bar l} \bar v| < ( \bar l+1) |\bar w| \leq (2M -1)|\bar w|\leq (2M -1) |w|.
\]
Therefore $\frac{2\cdot 3^k}{2 M - 1} < |w| < \frac{3^{k + 2}}{l}$, so $l < \frac{9}{2} (2 M - 1) < 9 M$.

It follows that $\ST_\Phi(u)$ is $(9 M)$-power-free.
\end{proof}

\begin{remark}
Analogous to the construction preceding Example~\ref{intersection set}, we construct the shift $Y_k$ using $H_k$.
We do not need $H_k$ to be a group since we have shown that $\Phi^{n\cdot 3^k} (u)$ is a concatenation of words that belong to $H_k$. Since \eqref{first claim} holds for each $k\geq 1$, we have $\pi(X_{\ST_\Phi(u)})=\cap_k Y_k$.
\end{remark}

\begin{example}\label{p=3 example part 5}
We continue with our running example, in particular from Example~\ref{p=3 example part 4}, where $p=3$, $\Phi$ is the cellular automaton with generating function $\phi(x) = x + 1$, and the initial condition is generated by the substitution $\theta(0)=001$, $\theta(1)=112$, $\theta(2)=220$. We saw that $H_k$, the group generated by $\{ \theta^k(0), \theta^k(1), \theta^k(2) \}$, is
\[ H_k = \{ 0^{3^k}, 1^{3^k}, 2^{3^k}, \theta^k(0), 2\theta^k(0), \theta^k(1), 2\theta^k(1), \theta^k(2), 2\theta^k(2)\}.\]
If we take $u= \cdots u_{-2}u_{-1}\cdot u_{0} u_{1} \cdots$ to be any bi-infinite fixed point of $\theta$, then $\ST_\Phi(u)$ is horizontally $M$-power-free for some $M$ by Theorem~\ref{power-free}.
\end{example}

In Theorem~\ref{power-free}, we fixed the cellular automaton and prime $p$, and we let $\theta$ vary over a family of substitutions.
Next, for each $p$ we fix a substitution and vary the cellular automaton to obtain nontrivial $(\sigma, \Phi)$-invariant measures for a family of cellular automata.

\begin{definition}
For fixed $p$, let $W \colonequal 01 \cdots (p-1)$ and define $\theta: \F_p\rightarrow \F_p^p$ by $\theta(a) = W + a^{p}$, where $a^p$ denotes the word $a a \cdots a$ of length $p$. We call $\theta$ the (base-$p$) {\em parity substitution}.
\end{definition}

If $u\in \F_p^\N$ is the fixed point of the parity substitution starting with $0$, then $u_m$ is the sum, modulo $p$, of the digits in the base-$p$ representation of $m$.

\begin{lemma}
The fixed point $u \in \F_p^\N$ of the parity substitution $\theta:\F_p\rightarrow \F_p^{p}$ is not eventually periodic.
\end{lemma}

\begin{proof}

For each candidate period length $k$, we show that there are arbitrarily large $m$ such that $u_m \neq u_{m+k}$.
Let $k_\ell \cdots k_1 k_0$ be the base-$p$ representation of $k$, with $k_\ell \neq 0$.
If $u_k \neq 0$, let $m = p^N$ for some $N > \ell$; then $u_m = 1 \nequiv 1 + u_k \equiv u_{m+k} \mod p$.
If $u_k = 0$, let $m = p^N + (p - k_\ell) p^\ell$ for some $N > \ell + 1$; then $u_m \equiv 1 + p - k_\ell \nequiv 2 - k_\ell \equiv u_{m+k} \mod p$.

\end{proof}

\begin{theorem}\label{power-free general}
Let $u\in \F_p^\Z$ be a fixed point of the parity substitution $\theta:\F_p\rightarrow \F_p^{p}$,
let $\Phi:\F_p^\Z\rightarrow \F_p^\Z$ be a linear cellular automaton,
and let $L$ be the number of nonzero monomials in the generating polynomial of $\Phi$.
If $p$ does not divide $L$, then there exists $M$ such that $\ST_\Phi(u)$ is horizontally $M$-power-free.
\end{theorem}

\begin{proof}
The proof is similar to that of Theorem~\ref{power-free}.
We refine Equation~\eqref{first claim} to claim that
\begin{equation}\label{first claim extra}
	\Phi^{n\cdot p^k}(u)|_{[mp^k, (m + 1) p^k - 1]} \in
		\{(L^n \bmod p) \theta^k(0) + a^p: a\in \F_p\}
\end{equation}
for each $n\in \N$. The proof of this claim is by induction, as in Theorem~\ref{power-free}. Note that $L^n \nequiv 0 \mod p$ for every $n$, since $p$ does not divide $L$.
Next we let
\[
	H_k = \{ j \theta^k(0)+ a^p : \text{$a\in \F_p$ and $j \equiv L^n \mod p$ for some $n\in \N$}\}.
\]
As in the proof of Theorem~\ref{power-free}, there exists $M \geq 2$ such that all words in $H_k$ are $M$-power-free.
Also, if $|w| \leq \frac{p^k}{M}$ and $w^l$ occurs as a subword of $W_1\cdots W_{j} \in H_k^{j}$, then $l \leq 2M - 2$.

Let $\ell$ and $r$ be the left and right radii of $\Phi$.
Given a word $w=w_1 \cdots w_m$ of length $m \geq 2$, define $\Phi(w)$ to be the word of length $m - \ell - r$ obtained by applying $\Phi$'s local rule.
Suppose $w^l$ occurs in the $n$-th row of $\ST_\Phi(u)$.
We show that
\[
	l \leq \max\left((\ell + r) p^2 M, \left\lceil \frac{p^2}{p - 1} (2 M - 1) \right\rceil\right).
\]
If $\frac{| w^l |}{\ell + r} < p$, then $l \leq l |w| < (\ell + r) p < (\ell + r) p^2 M$.
If $\frac{| w^l |}{\ell + r} \geq p$, let $k$ be such that $(\ell + r) p^{k + 1} \leq | w^l | = l |w| < (\ell + r) p^{k + 2}$.
Then $|w| < \frac{(\ell + r) p^{k + 2}}{l}$.
Let $N$ be such that ${N\cdot p^k}\leq n< (N+1)\cdot p^k$.
Write $\Phi^{ (N+1)\cdot p^k - n}(w^l)= \bar w^{\bar l} \bar v$, where the words $\bar w$ and $\bar v$ are such that $\bar l \geq 1$ is maximal and $\bar v$ is a prefix of $\bar w$ with $0\leq |\bar v|\leq |\bar w|-1$.
We have $|\bar w| \leq |w|$ since the period length of a word does not increase after applying $\Phi$.
There are two cases.

If $|\bar w| \geq \frac{p^k}{M}$, then $\frac{p^k}{M} \leq |\bar w| \leq |w| < \frac{(\ell + r) p^{k + 2}}{l}$, so $l < (\ell + r) p^2 M$.

If $|\bar w| < \frac{p^k}{M}$, then, since $\bar w^{\bar l}$ occurs on row $(N+1) \cdot p^k$, by \eqref{first claim extra} $\bar w^{\bar l}$ occurs as a subword of $W_1\cdots W_{j} \in H_k^{j}$ for some $j$.
By the same argument in the proof of Theorem~\ref{power-free}, $\bar w^{\bar l}$ occurs as a subword of $W_iW_{i+1}$ and therefore $\bar l \leq 2M-2$.
We also have
\begin{align*}
	|\bar w^{\bar l} \bar v| =
	|\Phi^{ (N+1)\cdot p^k - n}(w^l)|
	&= |w^l| - \left((N+1)\cdot p^k - n\right)(\ell + r) \\
	&\geq (\ell + r) p^{k + 1} - (N+1)\cdot p^k(\ell + r) + N \cdot p^k(\ell + r) \\
	&= (\ell + r) p^{k + 1} - p^k (\ell + r) \\
	&= (\ell + r) (p - 1) p^k
\end{align*}
so
\[
	(\ell + r) (p - 1) p^k \leq |\bar w^{\bar l} \bar v| < ( \bar l+1) |\bar w| \leq (2M -1)|\bar w|\leq (2M -1) |w|.
\]
Therefore $\frac{(\ell + r) (p - 1) p^k}{2 M - 1} < |w| < \frac{(\ell + r) p^{k + 2}}{l}$, so $l < \frac{p^2}{p - 1} (2 M - 1) \leq \left\lceil \frac{p^2}{p - 1} (2 M - 1) \right\rceil$.

It follows that $\ST_\Phi(u)$ is $\max\left((\ell + r) p^2 M, \left\lceil \frac{p^2}{p - 1} (2 M - 1) \right\rceil\right)$-power-free.
\end{proof}

\begin{question}
Given a linear cellular automaton $\Phi:\F_p^\Z\rightarrow \F_p^\Z$, what is the proportion of length-$p$ substitutions $\theta:\F_p\rightarrow \F_p^p$, with a bi-infinite $\theta$-fixed point $u$, for which there exists an $M \geq 2$ such that $\ST_\Phi(u)$ is horizontally $M$-power-free?
\end{question}

Einsiedler~\cite{Einsiedler-2004}, as well as finding the invariant sets that are discussed in Section~\ref{invariant sets}, shows the existence of shift-invariant measures supported on a subset of $X_\Phi$ (the set of spacetime diagrams for a linear cellular automaton $\Phi$).
He asks: What are the ergodic measures on $X$?
Our contribution is to identify simplices of invariant measures that are generated by ergodic measures supported on codings of substitutional sets. The invariant measures of a substitutional dynamical system can be derived from its incidence matrix: see \cite{Bezuglyi:2010} for a thorough description of how to compute them from the relevant Perron vectors of the matrix. The theory for two-dimensional substitutions is very similar and is described for primitive substitutions in \cite{Bartlett-2018}.

\begin{theorem}\label{nature of measures}
Let $\Phi:\F_p^\Z\rightarrow \F_p^\Z$ be a linear cellular automaton, and let $U \in \F_p^{\Z\times \Z}$ be a $[-p,-p]$-automatic spacetime diagram for $\Phi$.
Then there exists a simplex of $(\sigma_1,\sigma_2)$-invariant measures generated by the relevant Perron vectors of the incidence matrices of the four substitutions defining $U$.
\end{theorem}

\subsection{Automatic spacetime diagrams with finitely supported invariant measures}\label{Spacetime diagrams with atomic frequencies}

Given a length-$p$ substitution $\theta:\mathcal A \rightarrow \mathcal A^{*}$, recall that we write 
$\theta (a)= \theta_0(a) \cdots \theta_{p-1}(a)$, i.e.\ for $0\leq i \leq p-1$ we have a map
$\theta_i:\mathcal A \rightarrow \mathcal A$ where $\theta_i(a)$ is the 
$(i+1)$-st letter of $\theta(a)$.
We say that $\theta$ {\em has a coincidence} if 
there exists $k \geq 1$ and $i_1, \dots, i_k$ such that
\[
	|\theta_{i_1} \circ \dots \circ \theta_{i_k}(\mathcal A)| = 1.
\]
(The notion of a coincidence has dynamical significance, as a constant-length substitution with a coincidence defines a subshift which has discrete spectrum and so is measure theoretically a group rotation. There are various generalisations of the notion of a coincidence, such as the {\em strong coincidence condition}~\cite{Arnoux-Ito-2001} for non-constant-length substitutions; it is conjectured that a substitution satisfying the strong coincidence condition also has discrete spectrum.)
By considering a power of $\theta$ if necessary, we assume that the coincidence is achieved by $\theta$, i.e.\ $|\theta_i(\mathcal A)|=1$ for some $i$. Analogously, 
 we say that a $p$-automatic sequence $u$ has a coincidence if $u = \tau(\theta^\infty(a))$ for some length-$p$ substitution $\theta$ with a coincidence.
Given a word $w=w_0 w_1 \cdots w_n$, let $w_{[i,j)} \colonequal w_i w_{i+1} \cdots w_{j-1}$. 

Let $\Phi$ be a linear cellular automaton, let $ u \in \F_p^\Z$, and let $U=\ST_\Phi(u)$. Notice that $X_U$ contains the constant zero configuration if for all $N$ and $m$ there exists $n>N$ and $k$ such that $0^{m}$ occurs in the row $\Phi^n(u)$ starting at index $k$, as this implies that $\ST_\Phi(u)$ contains arbitrarily large triangles of $0$'s.
We investigate when $X_U$ contains constant configurations.

\begin{remark}\label{left radius 0}
In the following two theorems we assume that the cellular automaton $\Phi$ has left radius $0$. This is not a serious restriction for the following reason.
If $\Phi$ has generating polynomial $\phi(x)$ and has left radius $\ell$, then the generating polynomial $x^{-\ell}\phi(x)$ is the generating polynomial of a linear cellular automaton $\Psi$ with left radius $0$. Further, the $n$-th row of $\ST_\Psi(u)$ is the left shift, by $\ell n$ units, of the $n$-th row of $\ST_\Phi(u)$. In the case where $u_m = 0$ for $m \leq 0$, this tells us that the shears of $\ST_\Psi(u)$ and $\ST_\Phi(u)$ coincide.
By Theorem~\ref{Eilenberg shear}, the unsheared spacetime diagram $\ST_\Phi(u)$ has a finite $[p, p]$-kernel if and only if the sheared spacetime diagram $\ST_\Psi(u)$ is $[p, p]$-automatic.
\end{remark}

Note that Theorems~\ref{coincidence} and \ref{unique letter at nonzero coefficients} do not apply to the generating polynomial $\phi(x) = x + 1 \in \F_3[x]$ in Examples~\ref{p=3 example}, \ref{p=3 example part 2}, and \ref{p=3 example part 3} (even after shearing as in Remark~\ref{left radius 0}), since $\sum_{i = -\ell}^r \alpha_i \neq 0$.

\begin{theorem}\label{coincidence}
Let $u \in \F_p^\Z$ be such that $(u_m)_{m\geq 0}$ is $p$-automatic with a coincidence, and let $U=\ST_{\Phi}(u)$.
Let $\Phi:\F_p^\Z \rightarrow \F_p^\Z$ be a linear cellular automaton of left radius $0$ with generating polynomial $\phi(x) = \sum_{i=0}^r \alpha_i x^{-i} \in \F_p[x^{-1}]$.
If $\sum_{i=0}^{r}\alpha_i = 0$, then the constant zero configuration is an element of $X_U$.
\end{theorem}

\begin{proof}
Let $\theta:\mathcal A\rightarrow \mathcal A^{p}$ and $\tau:\mathcal A\rightarrow \F_p$ be the underlying substitution and coding defining $(u_m)_{m\geq 0}$.
Suppose first that $|\{\theta_{0}(a):a \in \mathcal A\}|=1$, i.e.\ that the coincidence is achieved in the leftmost column $\theta_0$, and also that the coincidence is attained by $\theta$. 
Thus there exists $a^*$ such that $\theta_0(a)=a^{*}$ for each $a\in \mathcal A$ and $ u_{np}=\tau (a^{*})$ for each $n\geq 0$. 
Since $ u$ is the coding of a $\theta$-fixed point, we have that $ u_{[ np^{j+1},np^{j+1}+p^j )} = \tau (\theta^j(a^{*}))$ for each $j\geq 0$ and each $n\geq 0$.

Since $\Phi^{p^\ell}$ has generating polynomial $ \sum_{i=0}^{r}\alpha_i x^{-ip^\ell}$, then
\[
	\Phi^{p^{j+1}}( u )_{ [ 0,p^j ) }
	= \sum_{i=0}^{r}\alpha_i u _{ [ ip^{j+1},ip^{j+1}+p^j)}
	= \sum_{i=0}^{r}\alpha_i \tau (\theta^j(a^{*} ))=0^{p^j},
\]
and in fact for each $m\geq 0$
\[
	\Phi^{p^{j+1}}(u )_{ [ mp^{j+1},mp^{j+1}+p^j ) }
	= \sum_{i=0}^{r}\alpha_i u _{ [ ip^{j+1}+mp^{j+1}, ip^{j+1}+mp^{j+1}+p^j)}
	= \sum_{i=0}^{r}\alpha_i \tau (\theta^j(a^{*} ))=0^{p^j}.
\] 

If the coincidence is achieved in the column $\theta_L$, we translate the above argument, starting with the modification that $ u_{np+L}=\tau (a^{*})$ for each $n\geq 0$, and adjusting accordingly.
\end{proof} 

\begin{example}\label{column_number_1_example}
Let $\theta$ be the substitution $\theta(a)=ab, \theta(b)=cd, \theta(c)=ac, \theta(d)=da$, and let $\tau(a)=\tau(c)=0, \tau(b)=\tau(d)=1$. Then $\theta^4$ has a coincidence in the 5-th column.
Let $u \colonequal \tau(\theta^{\infty}(a))$ and let $\phi(x) = 1 + x^{-1}$. Theorem~\ref{coincidence} tells us that $\ST_{\Phi}(u)$ contains arbitrarily large patches of $0$; see Figure~\ref{column_number_1_picture}.
The left half of the initial condition is the image under $\tau$ of the left-infinite fixed point of $\theta^3$ ending with $a$.
\end{example}

\begin{figure}
	\center{\includegraphics[scale=.75]{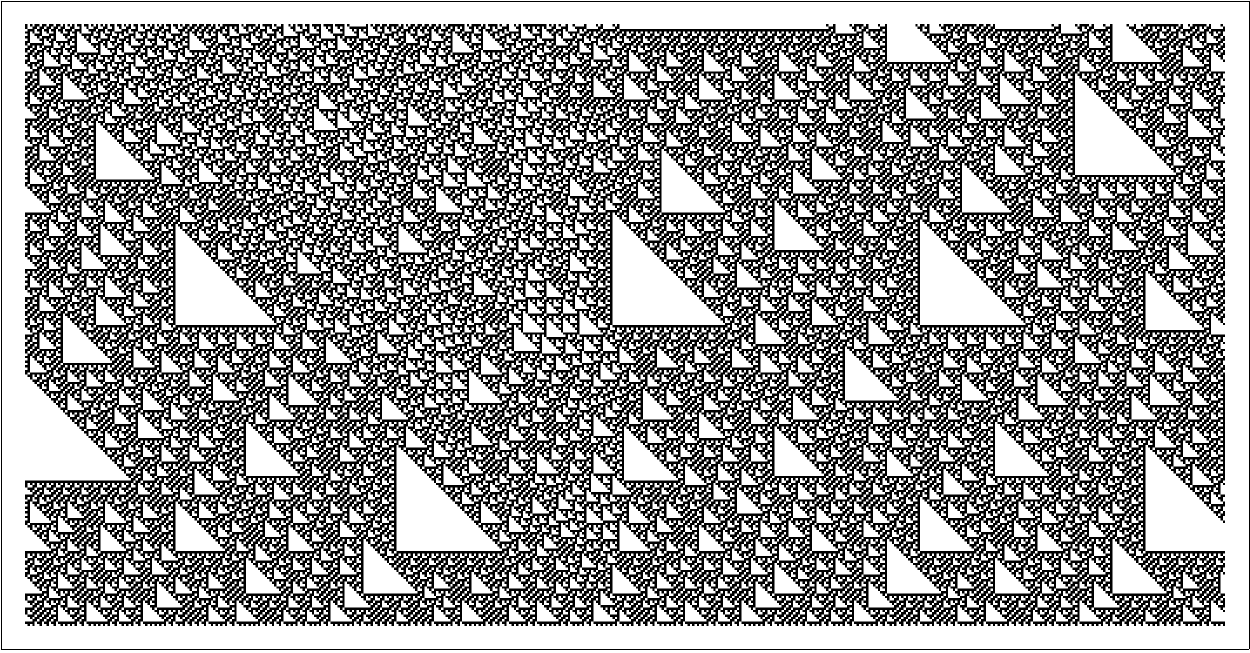}}
	\caption{Spacetime diagram for the Ledrappier cellular automaton, whose generating polynomial is $\phi(x) = 1 + x^{-1}$, with a $2$-automatic initial condition generated by the substitution in Example~\ref{column_number_1_example}.
	The dimensions are $511 \times 256$.}
	\label{column_number_1_picture}
\end{figure}

Substitutions with coincidences are not the only ones which generate shift spaces contain the constant zero configuration.
The next proposition identifies cellular automata and initial conditions which always give such a subshift.

\begin{theorem}\label{unique letter at nonzero coefficients}
Let $u \in \F_p^\Z$ be such that $(u_m)_{m\geq 0}$ is $p$-automatic, and let $U= \ST_{\Phi}(u)$.
Let $\theta : \mathcal A \to \mathcal A^p$ and $\tau : \mathcal A \to \F_p$ be such that $(u_m)_{m\geq 0} =\tau( \theta^\infty(a))$.
Let $\Phi : \F_p^\Z \rightarrow \F_p^\Z$ be a linear cellular automaton of left radius $0$ with generating polynomial $\phi(x) = \sum_{i = 0}^r \alpha_i x^{-i} \in \F_p[x^{-1}]$ such that $\sum_{i = 0}^r \alpha_i = 0$.
If there exists a finite word $w = w_0 w_1 \cdots w_r \in \mathcal A^{r + 1}$ such that $w$ occurs in $\theta^\infty(a)$ and $|\{w_i : \alpha_i \neq 0\}| = 1$, then $X_U$ contains the constant zero configuration.
\end{theorem}

\begin{proof}
Let $\{b\}=\{w_i:\alpha_i \neq 0 \}$. For each $j\geq 0$, since $w$ occurs in $\theta^\infty(a)$, then $\theta^j(w)$ also occurs in $\theta^\infty(a)$. Also, for each $i$ in the interval $0\leq i \leq r$ such that $\alpha_i \neq 0$, $\theta^j(b)$ occurs
at $\theta^j(w)_{[p^j i, p^j (i+1))}$. Since $\Phi^{p^j}$ has generating polynomial $\phi(x)^{p^j}= \sum_{i = 0}^r \alpha_i x^{-p^j i}$, we have, for each $k$ in the interval $0\leq k < p^j$,
\[ \left( \Phi^{p^j}\tau \left(\theta^j(w)\right) \right)_k = \sum_{i=0}^r \alpha_i \tau \left( \theta^j(w)_{p^j i +k} \right) = \sum_{i=0}^r \alpha_i \tau \left( \theta^j(b)_k\right) =\left(\sum_{i=0}^r \alpha_i\right) \tau \left(\theta^j(b)_k\right) =0,\]
so that the word $0^{p^j}$ occurs in $\ST_{\Phi}(u)$. The result follows.
\end{proof}

We remark that in the previous proof, it is sufficient that the word $w$ occurs once in $\theta^\infty(a)$, since for each $j$ we obtain a triangular region of $0$'s. Also, appropriate versions of the previous two theorems could be stated without left radius $0$; then we would also need to specify the left side of the initial condition.
Finally, given a $p$-automatic initial condition $u$, one can always find a linear cellular automaton $\Phi$ such that $\ST_{\Phi}(u)$ contains arbitrarily large words which are identically zero.
Conversely, given a linear cellular automaton $\Phi$ whose generating polynomial satisfies $\phi(1) = 0$, one can find an initial condition such that $\ST_{\Phi}(u)$ contains large words which are identically zero.
Theorems~\ref{coincidence} and \ref{unique letter at nonzero coefficients} are useful tools in Section~\ref{measure}, where we wished to avoid finitely supported invariant measures.

\begin{corollary}\label{p=2 Ledrappier}
Let $u \in \F_2^\Z$ be such that $(u_m)_{m\geq 0}$ is $2$-automatic, and let $U=\ST_\Phi(u)$.
Let $\Phi:\F_2^\Z \rightarrow \F_2^\Z$ be the Ledrappier cellular automaton with generating polynomial $\phi(x) = 1+x^{-1}\in \F_2[x^{-1}]$.
Then $X_U$ contains the constant zero configuration.
\end{corollary}

\begin{proof}
If $00$ or $11$ occurs in $(u_m)_{m\geq 0}$, we are done by Theorem~\ref{unique letter at nonzero coefficients}.
Otherwise, $(u_m)_{m\geq 0}$ is $0101\cdots$ or $1010\cdots$.
Since each of these sequences has a coincidence, we are done by Theorem~\ref{coincidence}.
\end{proof}

\begin{example}\label{Ledrappier example}
Let $\Phi$ be the Ledrappier cellular automaton
Let $\theta$ be the Thue--Morse substitution, $\theta(0)=01$ and $\theta(1)=10$, and let $p=2$. Then $00$ and $11$ occur in both fixed points of $\theta$ and the conditions of Corollary~\ref{p=2 Ledrappier} are satisfied; see Figure~\ref{Z x Z}.
\end{example}

\section*{Acknowledgement}

We thank Benjamin Hellouin de Menibus and Marcus Pivato for helpful discussions, and the referee for a careful reading.
Reem Yassawi thanks IRIF, Universit\'e Paris Diderot-Paris 7, for its hospitality and support.

\bibliographystyle{acm}
\bibliography{bibliography}

\end{document}